\documentclass[a4paper,11pt]{amsart}
\usepackage{amsmath,amssymb}
\usepackage{mathrsfs}
\usepackage{eucal}
\usepackage{cases}
\usepackage{comment}

\makeatletter
 
  \@addtoreset{equation}{section}
 \makeatother

\newcommand{\slim}{\operatorname{s-}\hspace{-0.25cm}\lim_{t\to\pm\infty}}

\newtheorem{thm}{Theorem}[section]
\newtheorem{lem}[thm]{Lemma}
\newtheorem{prop}[thm]{Proposition}

\theoremstyle{definition}

\newtheorem{ass}[thm]{Assumption}

\theoremstyle{remark}
\newtheorem{rem}[thm]{Remark}

\newtheorem{exs}[thm]{Examples}

\title[Long-range discrete Schr\"odinger]{Long-range scattering for discrete Schr\"odinger operators}
\author{Yukihide TADANO}
\thanks{
Graduate School of Mathematical Sciences, the University of Tokyo, 3-8-1 Komaba, Meguro, Tokyo, 153-8914, Japan \\
Research Fellow of Japan Society for the Promotion of Science \\
E-mail: {\tt tadano@ms.u-tokyo.ac.jp}}

\begin{document}
\begin{abstract}
In this paper, we define time-independent modifiers to construct a long-range scattering theory for a class of difference operators on $\mathbb{Z}^d$, including the discrete Schr\"odinger operators on the square lattice. The modifiers are constructed by observing the corresponding Hamilton flow on $T^*\mathbb{T}^d$. We prove the existence and completeness of modified wave operators in terms of the above mentioned time-independent modifiers.
\end{abstract}

\maketitle

\section{Introduction}

We consider a class of generalized discrete Schr\"odinger operators $H_0$ and $H$ on $\mathcal{H}=\ell^2(\mathbb{Z}^d)$, $d\geq1$,
\begin{align} \label{H_0,H}
\begin{cases}
H_0u[x]=\displaystyle\sum_{y\in\mathbb{Z}^d}f[y]u[x-y], \\
Hu[x]=H_0u[x]+V[x] u[x], 
\end{cases}
\end{align}
where $f \in \mathscr{S}(\mathbb{Z}^d):=\{u\in \ell^2(\mathbb{Z}^d) \mid u[x]=\mathcal{O}( \langle x \rangle^{-\infty}) \}$, $\langle x\rangle:=(1+|x|^2)^{\frac{1}{2}}$, satisfies $f[-x]=\overline{f[x]}$, $x \in \mathbb{Z}^d$, and  $V$ is a real-valued bounded function on $\mathbb{Z}^d$. Then $H_0$ and $H$ are bounded self-adjoint operators on $\mathcal{H}$.

We define the discrete Fourier transform $F$ by
\begin{align*}
Fu(\xi)=(2\pi)^{-\frac{d}{2}}\sum_{x\in\mathbb{Z}^d}e^{-ix\cdot \xi}u[x], \quad \xi \in \mathbb{T}^d=[-\pi,\pi)^d
\end{align*}
for $u\in \ell^1(\mathbb{Z}^d)$. Then $F$ is continuously extended to a unitary operator from $\mathcal{H}$ to $L^2(\mathbb{T}^d)$ and
\[
H_0u[x]=F^*\left(h_0(\cdot)Fu(\cdot)\right)[x],
\]
where
\begin{align}
h_0(\xi):=\sum_{x\in\mathbb{Z}^d}e^{-ix\cdot\xi}f[x], \quad \xi\in\mathbb{T}^d=[-\pi,\pi).
\end{align}
The above condition on $f$ implies $h_0$ is a real-valued smooth function on $\mathbb{T}^d$.
We denote by $v(\xi)$ and $A(\xi)$ the generalized velocity and the Hessian of $h_0$, respectively:
\begin{align*}
v(\xi)&=\nabla_\xi h_0(\xi), \\
A(\xi)&={}^t\nabla_\xi\nabla_\xi h_0(\xi)= (\partial_{\xi_j}\partial_{\xi_k} h_0(\xi))_{1\leq j,k \leq d}.
\end{align*}
The set of threshold energies is denoted by $\mathcal{T}$,
\begin{align*}
\mathcal{T}=\{ h_0(\xi) \mid \xi\in\mathbb{T}^d, v(\xi)=0 \} .
\end{align*}
We note $\mathcal{T}$ has Lebesgue measure $0$ by Sard's theorem.
We first assume the condition below.

\begin{ass}\label{ass1}
The sets $\{\xi\in\mathbb{T}^d \mid v(\xi) = 0 \}$ and $\{\xi\in\mathbb{T}^d \mid \det A(\xi)=0\}$ have $d$-dimensional Lebesgue measure zero.
\end{ass}

The above assumption implies the absence of point and singular continuous spectrum. The following assertion is a generalized version of Theorem 12.3.2 in \cite{O}.

\begin{prop}
Suppose that the set $\{\xi\in\mathbb{T}^d \mid v(\xi)=0\}$ has $d$-dimensional Lebesgue measure zero. Then $H_0$ has purely absolutely continuous spectrum and $\sigma_{\operatorname{ac}}(H_0)=h_0(\mathbb{T}^d)$, where $\sigma_{\operatorname{ac}}(H_0)$ denotes the absolutely continuous spectrum of $H_0$.
\end{prop}

\begin{proof}
Fix a point $\xi_0 \in W:=\{\xi\in\mathbb{T}^d \mid v(\xi)\neq0\}$. Then it suffices to prove $C_c^\infty(U) \subset \mathcal{H}_{ac}(F H_0 F^*)$ for some neighborhood $U \subset W$ of $\xi_0$; for any $f\in C_c^\infty(U)$,
\begin{align*}
\mathcal{B}(\sigma(H_0)) \to \mathbb{R},\ B \mapsto \int_{h_0^{-1}(B) \cap \operatorname{supp}u} |f(\xi)|^2 d\xi
\end{align*}
is an absolutely continuous Borel measure. The claim is proved by taking a local coordinate $U \ni x \mapsto (y(x),h_0(x)) \in \mathbb{R}^{d-1} \times \mathbb{R}.$
\end{proof}

If $V[x]$ decays at infinity, then $V$ is a compact operator on $\mathcal{H}$ and hence $
\sigma_{\operatorname{ess}}(H)=\sigma_{\operatorname{ess}}(H_0)=\sigma_{ac}(H_0)=h_0(\mathbb{T}^d)$,
where $\sigma_{\operatorname{ess}}(H)$ and $\sigma_{\operatorname{ess}}(H_0)$ denotes the essential spectrum of $H$ and $H_0$, respectively.
We suppose a long-range condition on $V$.


\begin{ass}\label{ass2}
There exist $\tilde V\in C^\infty({\mathbb{R}}^d;\mathbb{R})$ and $\varepsilon\in (0,1]$ such that $\tilde V\left|_{\mathbb{Z}^d}\right.=V$ and
\begin{align*}
|\partial_x^\alpha\tilde V(x)|\leq C_\alpha\langle x\rangle^{-|\alpha|-\varepsilon},& \quad x\in\mathbb{R}^d, \ \alpha\in\mathbb{Z}_{+}^d,
\end{align*}
where $\mathbb{Z}_{+}= \{ 0,1,2,\cdots \}$.
\end{ass}

Under Assumptions \ref{ass1} and \ref{ass2}, the singular continuous spectrum of $H$ is empty (see, e.g., \cite{N2}). In the following, we write $V$ for $\tilde V$ without confusion.

\begin{rem}
Assumption \ref{ass2} is equivalent to the following condition used in \cite{N},
\begin{align*}
|\tilde \partial_x^\alpha V\left[x\right]|\leq C_\alpha^\prime\langle x\rangle^{-|\alpha|-\varepsilon},& \quad x\in\mathbb{R}^d, \ \alpha\in\mathbb{Z}_{+}^d,
\end{align*}
where $\tilde \partial_x^\alpha=\tilde \partial_{x_1}^{\alpha_1}\cdots\tilde \partial_{x_d}^{\alpha_d}$, and $\tilde \partial_{x_j} V\left[x\right]=V\left[x\right]-V\left[x-e_j\right]$ is the difference operator with respect to the $j$-th variable. Here $\{e_j\}$ is the standard orthogonal basis of $\mathbb{R}^d$. See Lemma 2.1 in \cite{N} for the detail.
\end{rem}

In Section \ref{sec2}, we construct modified wave operators with time-independent modifiers, which are proposed by Isozaki and Kitada \cite{IsoKita}, so called Isozaki-Kitada modifiers. Isozaki-Kitada modifiers are formally defined by
\begin{align*}
W_J^\pm=\slim e^{itH}Je^{-itH_0}.
\end{align*}
We construct $J$ as an operator of the form
\begin{align} \label{J}
Ju\left[x\right]=(2\pi)^{-d}\int_{\mathbb{T}^d}\sum_{y\in\mathbb{Z}^d} e^{i(\varphi(x,\xi)-y\cdot\xi)}u[y] d\xi,
\end{align}
where the phase function $\varphi$ is a solution to the eikonal equation
\begin{align}\label{eikonal0}
h_0(\nabla_x \varphi(x,\xi))+V(x)=h_0(\xi)
\end{align}
in the ``outgoing'' and ``incoming'' regions and considered in Appendix \ref{appendix A}.

The next theorem is our main result.

\begin{thm}\label{main}
Under Assumptions \ref{ass1} and \ref{ass2}, there exists an operator $J$ of the form (\ref{J}) such that, for any $\Gamma\Subset h_0(\mathbb{T}^d)\backslash\mathcal{T}$, the modified wave operators
\begin{align}\label{mWO}
W_J^\pm(\Gamma):=\slim e^{itH}Je^{-itH_0}E_{H_0}(\Gamma)
\end{align}
exist, where $E_{H_0}$ denotes the spectral measure of $H_0$. Furthermore, the following properties hold:
\begin{enumerate}
\renewcommand{\labelenumi}{\roman{enumi})}
\item Intertwining property: $H W_J^\pm(\Gamma)=W_J^\pm(\Gamma) H_0$.
\item Partial isometries: $\|W_J^\pm(\Gamma) u\|=\|E_{H_0}(\Gamma)u\|$.
\item Asymptotic completeness: $\operatorname{Ran} W_J^\pm(\Gamma)=E_{H}(\Gamma)\mathcal{H}_{\operatorname{ac}}(H)$.
\end{enumerate}
\end{thm}

\begin{exs}
i) In \cite{N}, a long-range scattering theory of the standard difference Laplacian
$H_0u[x]=-\frac{1}{2}\sum_{|y-x|=1}u\left[y\right], \ x\in \mathbb{Z}^d$
is considered. In this case, $h_0(\xi)=-\sum_{j=1}^d \cos \xi_j$ satisfies Assumption \ref{ass1}.

ii) A model for $2$-dimensional triangle lattice is expressed by the operator
$H_0u[x]=-\frac{1}{6}\sum_{j=1}^6 u[x+n_j], \ x\in\mathbb{Z}^2,$
where $n_1=(1,0)$, $n_2=(-1,0)$, $n_3=(0,1)$, $n_4=(0,-1)$, $n_5=(1,-1)$, $n_6=(-1,1)$ (see, e.g., \cite{An-I-M}).
Since
\begin{align*}
h_0(\xi)=-\frac{1}{3}(\cos\xi_1+\cos\xi_2+\cos(\xi_1-\xi_2))
\end{align*}
in this case, Assumption \ref{ass1} is satisfied.
\end{exs}

Scattering theory for Schr\"odinger operators on $\mathbb{R}^d$ has been extensively studied (\cite{A-BdM-G}, \cite{D-G}, \cite{R-S}, \cite{Y}).
If the perturbation is long-range, i.e., $V(x)=O(\langle x \rangle^{-\varepsilon})$, $0<\varepsilon \leq 1$, then the scattering theory needs a modification (\cite{D-G}, \cite{IsoKita}, \cite{Y}).
Discrete Schr\"odinger operator describes the state of electrons in solid matters with graph structure. Spectral properties of discrete Schr\"odinger operators have been studied in \cite{An-I-M}, \cite{BdM-S}, \cite{I-Ko}, \cite{N}, \cite{N2}, \cite{P-R}.

The main idea of the construction of modifiers is similar to \cite{N}. We translate $H$ into an operator on the flat torus $\mathbb{T}^d$ via discrete Fourier transform and consider the corresponding classical mechanics on $\mathbb{T}^d$.
The proof is mainly based on \cite{IsoKita}. We use the time-decaying method to construct the phase function $\varphi$ in the definition of $J$, and then the stationary phase method and the Enss method to prove the existence and completeness of modified wave operators. The construction of $\varphi$ is given in Appendix \ref{appendix A}, which follows the argument of \cite{K}. The main properties of $\varphi$ is summarized in Proposition \ref{mainp}. In Section 2, we prepare some lemmas for the proof of Theorem \ref{main}. The Poisson summation formula is used to prove that pseudo-difference operators on $\mathbb{Z}^d$ are translated to pseudo-differential operators on $\mathbb{T}^d$ modulo smoothing operators (see the proof of Lemma \ref{lem2} in Appendix \ref{proof of lem}). This enables us to get over the difficulty derived from the discreteness of $\mathbb{Z}^d$.
 In Section 3, we prove Theorem \ref{main}.

\section*{acknowledgement}
The author would like to thank Professor Shu Nakamura, my Ph.D.\ advisor. This paper would not be completed without his advice.
The author is also grateful to Professor Hiroshi Isozaki for his kind discussion.

\section{Preliminaries} \label{sec2}
We first state a proposition on the Hamilton flow generated by $h(x,\xi):=h_0(\xi)+V(x)$, which is proved in Appendix \ref{appendix A}. Here we note that $h_0$, $v$ and $A$ are extended periodically in $\xi$ from $\mathbb{T}^d=[-\pi, \pi)^d$ to $\mathbb{R}^d$, and we identify integrations on $\mathbb{T}^d$ with those on $[-\pi, \pi)^d$.
We also note that the following proposition concerns functions on $\mathbb{R}^d\times\left(\mathbb{R}^d \backslash v^{-1}(0)\right)$, not on $\mathbb{Z}^d\times\left(\mathbb{T}^d \backslash v^{-1}(0)\right)$.

We fix $\chi\in C^\infty(\mathbb{R}^d)$ such that
\begin{align} \label{chi}
\chi(x)=
	\begin{cases}
	0 & \text{if } |x|\leq1 ,\\
	1 & \text{if } |x|\geq2,
	\end{cases}
\end{align}
and we define $\cos(x,y):=\frac{x \cdot y}{|x| |y|}$ for $x,y\in\mathbb{R}^d\backslash \{0\}$. The following assertion is an analogue of Theorem 2.5 in \cite{IsoKita}.

\begin{prop}\label{mainp}
There exists a real-valued function $\varphi\in C^\infty(\mathbb{R}^d\times(\mathbb{R}^d\backslash v^{-1}(0)))$ satisfying the following properties: Set $a>0$. Let $\varphi_a\in C^\infty(\mathbb{R}^d\times\mathbb{R}^d)$ be defined by
\begin{align}
\varphi_a(x,\xi)=(\varphi(x,\xi)-x\cdot\xi)\chi\left(\frac{v(\xi)}{a}\right)+x\cdot\xi .
\end{align}
$(1)$ 
The function $\varphi_a$ satisfies
\begin{align} \label{shuuki}
\varphi_a(x,\xi+2\pi m)&=\varphi_a(x,\xi)+2\pi x\cdot m, \quad m\in\mathbb{Z}^d, \\
|\partial_x^\alpha\partial_\xi^\beta\left[\varphi_a(x,\xi)-x\cdot\xi\right]|
&\leq C_{\alpha\beta,a} \langle x\rangle^{1-\varepsilon-|\alpha|}, \label{phase estimate} \\
|{}^t\nabla_x\nabla_\xi\varphi_a(x,\xi)-I|&<\frac{1}{2} \label{phase estimate2}
\end{align}
for $(x,\xi)\in\mathbb{R}^d\times\mathbb{R}^d$, where $|M|:=\left(\sum_{j,k=1}^d |M_{jk}|^2\right)^\frac{1}{2}$ for a matrix $M$.

\noindent $(2)$ 
We set
\begin{align}\label{J_a}
J_a u\left[x\right]:=(2\pi)^{-d}\int_{\mathbb{T}^d}\sum_{y\in\mathbb{Z}^d} e^{i(\varphi_a(x,\xi)-y\cdot\xi)}u[y]d\xi.
\end{align}
Then
\begin{align} \label{HJ_a-J_aH_0}
(HJ_a-J_aH_0)u\left[x\right]=(2\pi)^{-d}\int_{\mathbb{T}^d}\sum_{y\in\mathbb{Z}^d}e^{i(\varphi_a(x,\xi)-y\cdot\xi)}s_a(x,\xi)u\left[y\right]d\xi,
\end{align}
where
\begin{align}\label{s_a def}
s_a(x,\xi)&:=e^{-i\varphi_a(x,\xi)}H(e^{i\varphi_a(\cdot,\xi)})\left[x\right]-h_0(\xi) \\
\nonumber &=\sum_{z\in\mathbb{Z}^d} f[z] e^{i(\varphi_a(x-z,\xi)-\varphi_a(x,\xi))}+V\left[x\right]-h_0(\xi)
\end{align}
satisfies for $|x|\geq1$ and $|v(\xi)|\geq a$
\begin{align} \label{s_a}
|\partial_\xi^\beta s_a(x,\xi)|\leq \begin{cases}
C_{\beta,a}\langle x\rangle^{-1-\varepsilon}, & |\cos(x,v(\xi))|\geq \frac{1}{2}, \\
C_{\beta,a}\langle x\rangle^{-\varepsilon}, & |\cos(x,v(\xi))|\leq \frac{1}{2}.
\end{cases}
\end{align}
\end{prop}

We note that $\varphi_a$ satisfies the eikonal equation (\ref{eikonal0}) on $\{ (x,\xi) \mid |x|\geq R_a , |v(\xi)| \geq a , |\cos(x,v(\xi))|\geq \frac{1}{2} \}$ and that the property is used for the proof of (\ref{s_a}) in the $|\cos(x,v(\xi))|\geq \frac{1}{2}$ case (see Proposition \ref{eikonal prop} and (\ref{s2 estimate}).

In the rest of this section, we prepare some lemmas for the proof of properties ii) and iii).
We choose $\gamma\in C_c^\infty(h_0(\mathbb{T}^d)\backslash\mathcal{T})$ and $\rho_\pm\in C^\infty([-1,1];[0,1])$ such that
\begin{align*}
&\rho_+(\sigma)+\rho_-(\sigma)=1, \\
&\rho_+(\sigma)=1, \quad \sigma\in\left[\frac{1}{4},1\right], \\
&\rho_-(\sigma)=1, \quad \sigma\in\left[-1,-\frac{1}{4}\right].
\end{align*}
Using $\gamma$ and $\rho_\pm$, we define operators with cutoffs in the energy and the direction of $x$ and $v(\xi)$.
We set symbols $p_\pm$ and operators $P_\pm$, $\tilde P_\pm$ and $E_\pm(t)$ by
\begin{align}
p_\pm(y,\xi)&=\gamma(h_0(\xi))\chi(y)\rho_\pm(\cos(y,v(\xi))), \label{p+-} \\
P_\pm u\left[x\right]&=(2\pi)^{-d}\int_{\mathbb{T}^d}\sum_{y\in\mathbb{Z}^d} e^{i(x-y)\cdot\xi} p_\pm(y,\xi)u\left[y\right]d\xi, \\
\tilde P_\pm u\left[x\right]&=(2\pi)^{-d}\int_{\mathbb{T}^d}\sum_{y\in\mathbb{Z}^d} e^{i(x\cdot\xi-\varphi_a(y,\xi))}p_\pm(y,\xi)u\left[y\right]d\xi , \\
E_\pm(t)&=J_a e^{-itH_0}\tilde P_\pm , \quad t\in \mathbb{R}, \label{E_pm}
\end{align}
where $J_a$ is defined by (\ref{J_a}).
%

We consider properties of these operators.
We use the stationary phase method as in the pseudo-differential operator calculus (see, e.g., \cite{Z}).
The following two Lemmas correspond to Proposition 3.4 and Lemma 3.7 in \cite{IsoKita}, and
the proofs are given in Appendix \ref{proof of lem} (see also \cite{As-F} and \cite{IsoKita}).

\begin{lem}\label{lem}
$J_a$, $P_\pm$ and $\tilde P_\pm$ are bounded operators on $\mathcal{H}$.
\end{lem}


\begin{lem}\label{lem2}
$\gamma(H_0) - P_+ - P_-$, $P_\pm^*-P_\pm$, $E_\pm(0)-P_\pm$, $J_a^*J_a-I$ and $J_a J_a^*-I$ are compact operators on $\mathcal{H}$.
\end{lem}

The next lemma, corresponding to Proposition 3.8 in \cite{IsoKita}, is an analogue of the intertwining property of wave operators.

\begin{lem}\label{enssesti}
For any $s\in\mathbb{R}$,
\begin{align}
\slim e^{itH_0} J_a^* E_\pm(t-s)=e^{isH_0} \tilde P_\pm.
\end{align}
\end{lem}

\begin{proof}
The definition of $E_\pm(t)$ implies
\begin{align*}
e^{itH_0} J_a^* E_\pm(t-s)&=e^{itH_0} J_a^* J_a e^{-i(t-s)H_0} \tilde P_\pm \\
&=e^{itH_0} (J_a^* J_a-I) e^{-itH_0}e^{isH_0} \tilde P_\pm+e^{isH_0} \tilde P_\pm.
\end{align*}
Since $e^{-itH_0} u \to 0$ weakly as $t\to \pm\infty$ for any $u\in\mathcal{H}=\mathcal{H}_{\operatorname{ac}}(H_0)$, Lemma \ref{lem2}
implies that the first term converges strongly to 0 as $t\to\pm\infty$.
\end{proof}

Next we prove the norm convergence of $\lim_{t\to\pm\infty} e^{itH}E_\pm(t)$. If we set
\begin{align*}
G_\pm(t):=(\frac{d}{i d t}+H)E_\pm(t)=(H J_a - J_a H_0)E_\pm(t) ,
\end{align*}
then we have
\begin{align*}
e^{itH}E_\pm(t)-P_\pm=E_\pm(0)-P_\pm+i\int_0^t e^{i\tau H} G_\pm(\tau) d\tau .
\end{align*}
The following proposition is analogous to Theorem 3.5 in \cite{IsoKita}, and proves $G_\pm(t)$ is integrable in $\{ \pm t \geq 0\}$, respectively.

\begin{prop}\label{enss prop}
$G_\pm(t)$ is norm continuous and compact for any $t\in\mathbb{R}$. Furthermore, $G_\pm(t)$ satisfies
\begin{align}
\|G_\pm(t)\|\leq C\langle t\rangle^{-1-\varepsilon}, \ \pm t\geq 0. \label{G estimate}
\end{align}
In particular, $e^{itH}E_\pm(t)-P_\pm$ converges to a compact operator with respect to the norm topology as $t\to\pm\infty$, respectively.
\end{prop}

\begin{proof}
Let
\begin{align*}
\Phi(x,y,\xi;t):=\varphi_a(x,\xi)-th_0(\xi)-\varphi_a(y,\xi).
\end{align*}
Then the definition (\ref{E_pm}) of $E_\pm(t)$ implies
\begin{align*}
G_\pm(t)u[x]&=(HJ_a-J_aH_0)e^{-itH_0}\tilde P_\pm u[x] \\
&=(2\pi)^{-d} \int_{\mathbb{T}^d} \sum_{y\in\mathbb{Z}^d} e^{i\Phi(x,y,\xi;t)} s_a(x,\xi) p_\pm(y,\xi) u[y] d\xi.
\end{align*}
The norm continuity of $G_\pm(t)$ is obvious.
Furthermore, (\ref{s_a}) implies the compactness of $HJ_a-J_aH_0$ by the similar argument in the proof of Lemma \ref{lem2}, hence $G_\pm(t)$ is compact.

Let us prove (\ref{G estimate}). We consider the $+$ case only. The other case is proved similarly.
We use another decomposition $\rho^\pm \in C^\infty([-1,1]; [0,1])$ which is different from $\rho_\pm$ in that
\begin{align*}
\rho^+(\sigma)+\rho^-(\sigma)=1, \\
\rho^+(\sigma) =
\begin{cases}
1, &\sigma\geq \frac{3}{4},\\
0, &\sigma\leq \frac{1}{2}.
\end{cases}
\end{align*}
We define
\begin{align*}
s_-(x,\xi)&:=s_a(x,\xi)\chi_{\{x\neq0\}}\rho^-(\cos(x,v(\xi))), \\
s_+(x,\xi)&:=s_a(x,\xi)-s_-(x,\xi).
\end{align*}
We then decompose $G_+$ as
\begin{align} \label{G decomposition}
G_+(t)u[x]&=(2\pi)^{-d} \int_{\mathbb{T}^d} \sum_{y\in\mathbb{Z}^d} e^{i\Phi(x,y,\xi;t)} (s_+ p_+ + s_- p_+)(x,y,\xi)  u[y] d\xi \\
&=:(F_+(t)+F_-(t))u[x]. \nonumber
\end{align}
Now we claim that for any $t \geq 0$ and $\ell \geq 0$,
\begin{align}
&\| F_+(t) \| \leq C \langle at \rangle^{-1-\varepsilon}, \label{F_+} \\
&\|F_-(t)\|\leq C_\ell \langle at \rangle^{-\ell} . \label{F_-}
\end{align}
If (\ref{F_+}) and (\ref{F_-}) hold, then (\ref{G estimate}) follows from (\ref{G decomposition}).

For the proof of (\ref{F_+}), we let
\begin{align*}
\phi(t;y,\xi):=th_0(\xi)+\varphi_a(y,\xi)
\end{align*}
and set
\begin{align*}
L_1:=\langle \nabla_\xi \phi\rangle^{-2}(1-\nabla_\xi \phi\cdot D_\xi) .
\end{align*}
Then (\ref{phase estimate}) implies on the support of $s_+(x,\xi)p_+(y,\xi)$,
\begin{align*}
\langle \nabla_\xi \phi \rangle^{-1}\leq C \langle |y|+t|v(\xi)|\rangle^{-1} .
\end{align*}
Thus, for any $\ell \in \mathbb{Z}_+$, we have
\begin{align*}
F_+(t)u[x]
&=(2\pi)^{-d} \int_{\mathbb{T}^d} \sum_{y\in\mathbb{Z}^d} L_1^\ell (e^{-i\phi(t;y,\xi)})e^{i\varphi_a(x,\xi)} s_+(x,\xi)p_+(y,\xi) u[y] d\xi \nonumber \\
&=(2\pi)^{-d} \int_{\mathbb{T}^d} \sum_{y\in\mathbb{Z}^d} e^{-i\phi(t;y,\xi)}({}^tL_1)^\ell (e^{i\varphi_a(x,\xi)} s_+(x,\xi)p_+(y,\xi)) u[y] d\xi \nonumber \\
&=(2\pi)^{-d} \int_{\mathbb{T}^d} \sum_{y\in\mathbb{Z}^d} e^{i\Phi(t;y,\xi)}\{e^{-i\varphi_a(x,\xi)}({}^tL_1)^\ell (e^{i\varphi_a(x,\xi)} s_+ p_+)\} u[y] d\xi . \nonumber
\end{align*}
The function in $\{\}$ is a finite sum of functions of the form $s_j^\ell(x,\xi) p_j^\ell(y,\xi;t)$ such that
\begin{align} \label{decomposition}
\begin{cases}
|\partial_\xi^\beta s_j^\ell(x,\xi)|\leq C_\beta \langle x \rangle^{\ell-1-\varepsilon} , \\
|\partial_\xi^\beta p_j^\ell(y,\xi;t)|\leq C_\beta \langle |y|+t|v(\xi)| \rangle^{-\ell}.
\end{cases}
\end{align}
Indeed, (\ref{decomposition}) follows from (\ref{s_a}) and (\ref{p+-}).
Letting
\begin{align*}
S_j^\ell u[x]&:=(2\pi)^{-d} \int_{\mathbb{T}^d} \sum_{y\in\mathbb{Z}^d} e^{i(\varphi_a(x,\xi)-y\cdot\xi)} s_j^\ell(x,\xi) u[y] d\xi , \\
P_j^\ell(t) u[x]&:=(2\pi)^{-d} \int_{\mathbb{T}^d} \sum_{y\in\mathbb{Z}^d} e^{i(x\cdot\xi-\varphi_a(y,\xi))} p_j^\ell(y,\xi;t)  u[y] d\xi ,
\end{align*}
we have
\begin{align*}
F_+(t)=\sum_j S_j^\ell e^{-itH_0} P_j^\ell(t).
\end{align*}
Furthermore, we have by (\ref{decomposition}) and the argument in the proof of Lemma \ref{lem}
\begin{align*}
&\| \langle x \rangle^{1+\varepsilon-\ell} S_j^\ell \| < \infty , \\
&\|P_j^\ell(t)\| \leq C_\ell \langle at \rangle^{-\ell} .
\end{align*}
Thus we obtain
\begin{align*}
\|\langle x \rangle^{1+\varepsilon-\ell} F_+(t)\| \leq C^\prime_\ell \langle at \rangle^{-\ell}
\end{align*}
for any $\ell \in \mathbb{Z}_{+}$. Interpolation with respect to $\ell$ implies (\ref{F_+}).

For the proof of (\ref{F_-}), we note on the support of $s_-(x,\xi)p_+(y,\xi)$,
\begin{align*}
\langle \nabla_\xi \Phi \rangle^{-1} \leq C \langle  |x-y|+t |v(\xi)|\rangle^{-1} .
\end{align*}
Letting
\begin{align*}
L_2:=\langle \nabla_\xi \Phi \rangle^{-2} (1+\nabla_\xi \Phi\cdot D_\xi),
\end{align*}
we have
\begin{align*}
F_-(t)u[x]&=(2\pi)^{-d} \int_{\mathbb{T}^d} \sum_{y\in\mathbb{Z}^d} e^{i\Phi(x,y,\xi;t)}({}^tL_2)^\ell (s_-(x,\xi)p_+(y,\xi)) u[y] d\xi \\
&=(2\pi)^{-d} \int_{\mathbb{T}^d} \sum_{y\in\mathbb{Z}^d} e^{i(\varphi_a(x,\xi)-\varphi_a(y,\xi))}e^{-it h_0(\xi)}({}^tL_2)^\ell (s_- p_+) u[y] d\xi 
\end{align*}
for any $\ell \in \mathbb{Z}_+$. Since
\begin{align*}
q^\ell(x,y,\xi;t):=e^{-it h_0(\xi)}({}^tL_2)^\ell (s_-(x,\xi)p_+(y,\xi))
\end{align*}
satisfies
\begin{align*}
|\partial_\xi^\beta q^\ell(x,y,\xi;t)|\leq C_{\ell,\beta} \langle tv(\xi) \rangle^{|\beta|-\ell}
\end{align*}
for any $\ell \in \mathbb{Z}_+$, we obtain (\ref{F_-}) by the argument in the proof of Lemma \ref{lem}.
\end{proof}

The next proposition claims that any particle in the energy $\Gamma$ does not stay in any bounded domain in $x$.

\begin{prop}
For any $R>0$ and $\ell\geq0$,
\begin{align}\label{teisoku}
\|\chi_{\{|x|< R\}}E_\pm (s)\|\leq C_{\ell,R} \langle s\rangle^{-\ell}, \quad \pm s\geq 0 .
\end{align}

\end{prop}

\begin{proof}
We prove (\ref{teisoku}) for the $+$ case only. We first note
\begin{align*}
E_+ (s)u[x]=(2\pi)^{-d} \int_{\mathbb{T}^d} \sum_{y\in\mathbb{Z}^d} e^{i\Phi(x,y,\xi;s)} p_+ (y,\xi) u[y] d\xi,
\end{align*}
where $\Phi(x,y,\xi;t)=\varphi_a(x,\xi)-th_0(\xi)-\varphi_a(y,\xi)$. We observe that on the support of $p_+(y,\xi)$,
\begin{align*}
|s v(\xi)+\nabla_\xi\varphi_a(y,\xi)|\geq c(|y|+s |v(\xi)|)
\end{align*}
for large $s$. Then, if $|x|\leq R$, we have for $s>0$ large enough
\begin{align*}
|\nabla_\xi \Phi(x,y,\xi;s)|\geq c(|y|+s |v(\xi)|), \quad (y,\xi)\in \operatorname{supp}p_+.
\end{align*}
Similarly to the proof of (\ref{F_-}), we obtain (\ref{teisoku}).
\end{proof}

\section{Proof of Theorem \ref{main}}

\subsection{Existence of modified wave operators}
We prove the existence of the limit (\ref{mWO}) for the $+$ case only. The other case is proved similarly. First we fix $\Gamma\Subset h_0(\mathbb{T}^d)\backslash\mathcal{T}$. 
We remark that, for any $u\in \mathcal{H}$ such that $Fu\in C^\infty(\mathbb{T}^d)$ and $\operatorname{supp}Fu\subset h_0^{-1}(\Gamma)$, we have
\begin{align} \label{JJa}
JE_{H_0}(\Gamma)u=J_a u
\end{align}
for some small enough $a>0$.
Then, to prove the existence of the limit (\ref{mWO}), it suffices to show that
\begin{align}\label{pf existence}
&\int_0^\infty \left\|\frac{d}{dt} \left(e^{itH}Je^{-itH_0}E_{H_0}(\Gamma)u \right)\right\| dt \\
&=\int_0^\infty \left\|\frac{d}{dt} \left(e^{itH}J_ae^{-itH_0}u \right)\right\| dt \nonumber \\
&= \int_0^\infty \|e^{itH}(HJ_a-J_aH_0)e^{-itH_0}u\| dt \nonumber \\
&= \int_0^\infty \|(HJ_a-J_aH_0)e^{-itH_0}u\| dt \nonumber
\end{align}
is finite for such $u$. The last equality follows from the fact that $e^{itH}$ is a unitary operator. Furthermore, by Assumption \ref{ass1} and a partition of unity on $\mathbb{T}^d$, we may assume that $Fu\in C^\infty(\mathbb{T}^d)$ has a sufficiently small support in $\{\xi\in h_0^{-1}(\Gamma) \mid \det A(\xi)\neq0\}$.

Let $w(t):= (HJ_a-J_aH_0)e^{-itH_0}u$. Then (\ref{HJ_a-J_aH_0}) implies
\begin{align*}
w(t)[x]=(2\pi)^{-\frac{d}{2}}\int_{\mathbb{T}^d} e^{i(\varphi_a(x,\xi)-th_0(\xi))} s_a(x,\xi) Fu(\xi) d\xi . \nonumber
\end{align*}
Now we use the stationary phase method. The stationary point $\xi=\xi(x,t)$ is determined by
\begin{align}
\frac{1}{t}\nabla_\xi \varphi_a(x,\xi)-v(\xi)=0. \label{stationary point}
\end{align}
We define
\begin{align*}
D_t:= \{x\in\mathbb{Z}^d \mid \exists \xi\in\operatorname{supp}Fu \text{ s.t.}\ (\ref{stationary point}) \text{ holds}\}.
\end{align*}
By (\ref{phase estimate}), there exists an open set $U\Subset \{\xi\in h_0^{-1}(\Gamma) \mid \det A(\xi)\neq0\}$ such that $\operatorname{supp}Fu \Subset U$ and that for $t>0$ large enough,
\begin{align*}
D_t \subset \left\{x \mid \frac{x}{t}\in v(U)\right\}=:D_t^\prime .
\end{align*}

On $(D_t^\prime)^c$, the non stationary phase method implies
\begin{align*}
|w(t)[x]|\leq C_\ell \langle |x|+t \rangle^{-\ell},\quad x\in\mathbb{Z}^d,\ t>0
\end{align*}
for any $\ell\geq0$. Thus we learn for any $\ell \geq 0$
\begin{align}\label{nonstationary}
\|\chi_{(D_t^\prime)^c}w(t)\| \leq C_\ell^\prime t^{-\ell}.
\end{align}

%

On $D_t^\prime$, the stationary phase method implies
\begin{align*}
w(t)[x]=t^{-\frac{d}{2}} A(t,x)s_a(x,\xi(x,t))Fu(\xi(x,t))+t^{-\frac{d}{2}-1}r(t,x),
\end{align*}
where $A(t,x)$ is uniformly bounded in $x$ and $t$ with $x\in D_t^\prime$, and
\begin{align*}
|r(t,x)|\leq C\sup_{|\beta|\leq d+3} \sup_{\xi\in\operatorname{supp}Fu} |\partial_\xi^\beta s_a(x,\xi)| .
\end{align*}
Since $\cos(x,v(\xi))\geq\frac{1}{2}$ for $x\in D_t^\prime$ and $\xi \in \operatorname{supp}Fu$ if $t$ is sufficiently large, we have by (\ref{s_a})
\begin{align*}
|s_a(x,\xi(x,t))|&\leq C\langle x\rangle^{-1-\varepsilon}, \\
|r(t,x)|&\leq C\langle x\rangle^{-1-\varepsilon} .
\end{align*}
We note $|x|\sim t$ on $D_t^\prime$ and the Lebesgue measure of $D_t^\prime$ is bounded by $C t^d$.
Thus we learn
\begin{align}\label{stationary}
\|\chi_{D_t^\prime}w(t)\| \leq \left(\int_{D_t^\prime} \left(C t^{-\frac{d}{2}} \langle x \rangle^{-1-\varepsilon}\right)^2 dx \right)^{\frac{1}{2}}
\leq C^\prime t^{-1-\varepsilon} .
\end{align}
Hence (\ref{nonstationary}) and (\ref{stationary}) imply
\begin{align*}
\|w(t)\|&\leq \|\chi_{D_t^\prime}w(t)\| + \|\chi_{(D_t^\prime)^c}w(t)\| \leq C^{\prime\prime} t^{-1-\varepsilon},
\end{align*}
which proves (\ref{pf existence}) is finite.
\hfill$\Box$

\subsection{Proof of the properties i), ii) and iii)}\label{subsec3.2}
\begin{proof}[Proof of i)]
The intertwining property is proved similarly to the short-range case (see, e.g., \cite{R-S}).
\end{proof}

\begin{proof}[Proof of ii)]
It suffices to show $\|W_J^\pm(\Gamma) u\|=\|u\|$ for $Fu\in C^\infty(\mathbb{T}^d)$ with $\operatorname{supp} Fu\subset h_0^{-1}(\Gamma)$. For such $u$, $Ju=J_a u$ holds for small $a>0$. Thus letting $u_t=e^{-itH_0}u$, we learn
\begin{align*}
\|W_J^\pm(\Gamma) u\|^2=\lim_{t\to\pm\infty} \|J_a u_t\|^2=\lim_{t\to\pm\infty} ((J_a^* J_a-I)u_t,u_t)+\|u\|^2 .
\end{align*}
Using $\operatorname{w-lim}_{t\to\pm\infty} u_t=0$ and Lemma \ref{lem2}, we have $\lim_{t\to\pm\infty} (J_a^* J_a-I)u_t=0$. This proves $W_J^\pm(\Gamma)$ are partial isometries.
\end{proof}

\begin{proof}[Proof of iii)]
We prove the asymptotic completeness of $W_J^+(\Gamma)$ only. Since intertwining property implies $\operatorname{Ran} W_J^+(\Gamma) \subset E_{H}(\Gamma)\mathcal{H}_{\operatorname{ac}}(H)$, it suffices to prove $\operatorname{Ran} W_J^+(\Gamma) \supset E_{H}(\Gamma)\mathcal{H}_{\operatorname{ac}}(H)$. 

We fix $v\in\mathcal{H}_{\operatorname{ac}}(H)$ and $\gamma\in C^\infty(\mathbb{R})$ so that $\gamma(H)v=v$ and $\operatorname{supp}\gamma\subset \Gamma$.
We set $v_t:=e^{-itH}v$ for simplicity.
Then we show that $E_{H}(\Gamma)\mathcal{H}_{\operatorname{ac}}(H) \subset \operatorname{Ran}W_J^+(\Gamma)$ follows from
\begin{align}\label{enss}
\lim_{s\to\infty}\limsup_{t\to\infty} \| v_s-e^{i(t-s)H}E_+(t-s)v_s \| =0 .
\end{align}
First, we observe
\begin{align*}
&\|e^{itH_0}J_a^* e^{-itH}v-e^{is H_0}\tilde P_+ v_s\| \\
&\leq \|e^{itH_0}J_a^*\left[v_t-E_+(t-s)v_s\right]\| + \| e^{itH_0} J_a^* E_+(t-s)v_s-e^{isH_0} \tilde P_+ v_s \| .
\end{align*}
Lemma \ref{enssesti} implies the second term tends to $0$ as $t\to\infty$. The first term is estimated by (\ref{enss}) since
\begin{align*}
&\|e^{itH_0}J_a^*\left[v_t-E_+(t-s)v_s\right]\| \\
&\leq \|e^{itH_0}J_a^*\| \|v_t-E_+(t-s)v_s\| \\
&=\|J_a^*\| \|e^{i(t-s)H}(v_t-E_+(t-s)v_s)\| \\
&=\|J_a^*\| \|v_s-e^{i(t-s)H}E_+(t-s)v_s\| .
\end{align*}
Thus we have
\begin{align*}
\lim_{s\to\infty}\limsup_{t\to\infty} \|e^{itH_0}J_a^* e^{-itH}v-e^{is H_0}\tilde P_+ v_s\|=0.
\end{align*}
This implies $\left\{ e^{itH_0}J_a^* e^{-itH}v \right\}_{t\geq0}$ is a Cauchy sequence in $\mathcal{H}$, equivalently, there exists the limit
\begin{align*}
\lim_{t\to\infty} e^{itH_0}J_a^* e^{-itH}v=: \Omega^a v.
\end{align*}
Hence we obtain for sufficiently small $a>0$,
\begin{align*}
v=W_{J}^+(\Gamma) \Omega^a v\in \operatorname{Ran} W_J^+(\Gamma) .
\end{align*}

In the rest of the proof, we show (\ref{enss}). Since $v_s=\gamma(H)v_s$, we have
\begin{align}
v_s-e^{i(t-s)H}E_+(t-s)v_s=&\gamma(H)v_s-e^{i(t-s)H}E_+(t-s)v_s \label{pf completeness} \\
=&(\gamma(H)-\gamma(H_0))v_s \nonumber \\
&+(\gamma(H_0)-P_+-P_-)v_s \nonumber \\
&+(P_+-e^{i(t-s)H}E_+(t-s))v_s+P_-v_s . \nonumber
\end{align}
We note $\operatorname{w-lim}_{s\to\infty}v_s=0$ and $\gamma(H)-\gamma(H_0)$ is compact by the compactness of $H-H_0= V$. We also note $\gamma(H_0)-P_+-P_-$ is compact by Lemma \ref{lem2}, and $P_+ -e^{i(t-s)H}E_+(t-s)$ converges to a compact operator independent of $s$ as $t\to\infty$ by Proposition \ref{enss prop}.
Thus the terms on the RHS of (\ref{pf completeness}) except the last one converge to 0.

To estimate the last term of (\ref{pf completeness}), we observe
\begin{align}\label{pf completeness2}
\|P_-v_s\|^2=&(P_-^*P_-v_s, v_s) \\
=&((P_-^*-P_-)P_-v_s, v_s) \nonumber \\
&+((P_--e^{-isH}E_-(-s))P_-v_s, v_s) \nonumber \\
&+(P_-v_s,E_-(-s)^*v) . \nonumber
\end{align}
By the similar argument as above, we learn the first and second terms of (\ref{pf completeness2}) converge to 0 as $s\to\infty$. The third term of (\ref{pf completeness2}) is bounded by
\begin{align} \label{R de bunkai}
&|(P_-v_s,E_-(-s)^*v)| \\
&=|(P_-v_s,E_-(-s)^*(\chi_{\{|x|\geq R\}}+\chi_{\{|x|<R\}})v)| \nonumber \\
&\leq \|E_-(-s) P_- v_s\| \|\chi_{\{|x|\geq R\}}v\| + \|P_- v_s\| \|\chi_{\{|x|< R\}}E_-(-s)\| \| v\| \nonumber \\
&\leq C_v (\|\chi_{\{|x|\geq R\}}v\| + \|\chi_{\{|x|< R\}}E_-(-s)\|) \nonumber
\end{align}
for any $R>0$. 
Using (\ref{teisoku}) and $\lim_{R\to\infty}\|\chi_{\{|x|\geq R\}}v\|=0$, we learn that (\ref{R de bunkai}) converges to $0$ as $s\to\infty$. Hence we obtain (\ref{enss}).
\end{proof}

\appendix
\section{Classical mechanics and the construction of phase function}
\label{appendix A}

In this appendix, we use the following notations: 
For $\rho\in (0,1)$, we define
\begin{align*}
h(x,\xi)=&h_0(\xi)+V(x), \\
V_\rho(t,x)=&V(x)\chi(\rho x)\chi\left(\frac{\langle\log\langle t \rangle\rangle x}{\langle t \rangle}\right), \\
h_\rho(t,x,\xi)=&h_0(\xi)+V_\rho(t,x), \\
\nabla_x^2 V_\rho(t,x)=&{}^t\nabla_x\nabla_x V_\rho(t,x),
\end{align*}
where $\chi\in C^\infty(\mathbb{R}^d)$ is a fixed function satisfying (\ref{chi}).
Let $\varepsilon$ be as in Assumption \ref{ass2}. We fix $\varepsilon_0,\varepsilon_1 > 0$ such that $\varepsilon_0+\varepsilon_1<\varepsilon$.

The construction of time-decaying potential is same as Isozaki and Kitada \cite{IsoKita}, and is first used by Kitada and Yajima \cite{K-Y}.
One of the merits of this construction is that $V_\rho$ decays with respect to time $t$ almost same as position $x$.
The next lemma follows from Assumption \ref{ass2} with elementary computations.

\begin{lem}\label{V time-decay}
For any $t\in\mathbb{R}$, $x\in\mathbb{R}^d$ and multi-index $\alpha$,
\begin{align} \label{Vrho}
|\partial_x^\alpha V_\rho(t,x)| \leq C_\alpha \min\{\rho^{\varepsilon_0}\langle t\rangle^{-|\alpha|-\varepsilon_1}, \langle x\rangle^{-|\alpha|-\varepsilon}\} ,
\end{align}
where $C_\alpha$'s are independent of $x$, $t$ and $\rho$.
\end{lem}

Let $(q,p)(t,s)=(q,p)(t,s;x,\xi)$ be the solution to the canonical equation associated to the Hamiltonian $h_\rho$:
\begin{align*}
\left\{
\begin{array}{l}
\partial_t q(t,s)=\nabla_{\xi} h_\rho(t,p(t,s),q(t,s)), \\
\partial_t p(t,s)=-\nabla_x h_\rho(t,p(t,s),q(t,s)), \\
(q,p)(s,s)=(x,\xi).
\end{array}
\right. 
\end{align*}
This can be rewritten in the integral form:
\begin{align}
q(t,s)
&=x+\int_s^t v(p(\tau,s)) d\tau , \label{canq}\\
p(t,s)
&=\xi-\int_s^t \nabla_x V_\rho(\tau,q(\tau,s)) d\tau . \label{canp}
\end{align}

Before proving Proposition \ref{mainp}, let us describe the outline of this section.
First, we see in Proposition \ref{general estimates} that $q(t,s) \sim x + (t-s)v(\xi)$ and $p(t,s) \sim \xi$ for sufficiently small $\rho>0$.
Then we construct a solution $\phi(t;x,\xi)$ of the Hamilton-Jacobi equation (\ref{Hamilton-Jacobi}) by the method of characteristics.
Also estimates for $y(s,t;x,\xi)$ and $\eta(t,s;x,\xi)$, characterized by (\ref{y}) and (\ref{eta}), respectively, are given in Proposition \ref{y eta estimates}.
Using the above $\phi$, we define functions $\phi_\pm(x,\xi)$ by (\ref{phi definition}), and we confirm that $\phi_\pm$ satisfies the eikonal equation (\ref{eikonal0}) and 
the estimate (\ref{phase estimate}) in outgoing and incoming region, respectively.
Finally, we construct a function $\varphi(x,\xi)$ such that Proposition \ref{mainp} holds with $\phi_\pm$ and phase-space cutoffs.

Now, we start with estimates for classical orbits $(q,p)(t,s;x,\xi)$.
The following proposition is the corresponding result of Proposition 2.1 in \cite{IsoKita}.

\begin{prop}\label{general estimates}
For $\rho>0$ small enough, there exist $C_\ell>0\ (\ell\in \mathbb{Z}_{+})$ such that, for any $x,\xi\in\mathbb{R}^d$, $0\leq\pm s \leq\pm t $ and multi-indices $\alpha$ and $\beta$,
\begin{align}
&|p(s,t;x,\xi)-\xi|\leq C_0\rho^{\varepsilon_0}\langle s\rangle^{-\varepsilon_1}, \label{p1}\\
&|p(t,s;x,\xi)-\xi|\leq C_0\rho^{\varepsilon_0}\langle s\rangle^{-\varepsilon_1}, \label{p2}\\
&|\partial_x^\alpha\left[\nabla_x q(s,t;x,\xi)-I \right]|\leq C_{|\alpha|}\rho^{\varepsilon_0}\langle s\rangle^{-\varepsilon_1}, \label{qs}\\ 
&|\partial_x^\alpha\nabla_x p(s,t;x,\xi)|\leq C_{|\alpha|}\rho^{\varepsilon_0}\langle s\rangle^{-1-\varepsilon_1}, \label{ps}\\
&|\partial_x^\alpha\partial_\xi^\beta\left[\nabla_x q(t,s;x,\xi)-I \right]|\leq C_{|\alpha|+|\beta|}\rho^{\varepsilon_0}\langle s\rangle^{-1-\varepsilon_1}|t-s|,\label{qx} \\ 
&|\partial_x^\alpha\partial_\xi^\beta\nabla_x p(t,s;x,\xi)|\leq C_{|\alpha|+|\beta|}\rho^{\varepsilon_0}\langle s\rangle^{-1-\varepsilon_1}, \label{px}\\
&|\partial_\xi^\beta\left[\nabla_\xi q(t,s;x,\xi)-(t-s)A(\xi) \right]|\leq C_{|\beta|}\rho^{\varepsilon_0}\langle s\rangle^{-\varepsilon_1}|t-s|,\label{qxi} \\ 
&|\partial_\xi^\beta\left[\nabla_\xi p(t,s;x,\xi)-I\right]|\leq C_{|\beta|}\rho^{\varepsilon_0}\langle s\rangle^{-\varepsilon_1},\label{pxi} \\
&|\partial_x^\alpha\partial_\xi^\beta\left[q(t,s;x,\xi)-x-(t-s)v(p(t,s;x,\xi))\right]| \label{qesti} \\
& \hspace{2cm}\leq C_{|\alpha|+|\beta|}\rho^{\varepsilon_0} \min\{|t-s|\langle s\rangle^{-\varepsilon_1},\langle t\rangle^{1-\varepsilon_1}\}. \nonumber
\end{align}
Here, $|x|=\left(\sum_{j=1}^{d}|x_j|^2\right)^{\frac{1}{2}}$ for a vector $x$ and $|M|=\left(\sum_{j,k=1}^d |M_{jk}|^2\right)^\frac{1}{2}$ for a matrix $M$.
\end{prop}

\begin{proof}
We prove in the $0\leq s \leq t$ case. The other case is proved similarly.
The proof is decomposed into 5 steps.

\underline{Step 1}: Proof of (\ref{p1}) and (\ref{p2}). The inequalities (\ref{p1}) and (\ref{p2}) are shown by (\ref{Vrho}) and
\begin{align*}
p(t,t^\prime)-\xi=-\int_{t^\prime}^t \nabla_x V_\rho(\tau,q(\tau,t^\prime)) d\tau,
\quad t,t^\prime\in\mathbb{R}.
\end{align*}

\underline{Step 2}: Proof of (\ref{qs}) and (\ref{ps}). We use the induction with respect to $|\alpha|$.
First we prove (\ref{qs}) and (\ref{ps}) for $\alpha=0$.
Differentiating (\ref{canq}) and (\ref{canp}) in $x$, we have
\begin{align*}
\left\{
\begin{array}{l}
\nabla_x q(s,t)=I+\int_t^s A(p(\tau,t))\nabla_x p(\tau,t) d\tau, \\
\nabla_x p(s,t)=-\int_t^s \nabla_x^2 V_\rho(\tau,q(\tau,t))\nabla_x q(\tau,t) d\tau.
\end{array}
\right.
\end{align*}
Letting
\begin{align*}
Q_0(s)&:=\nabla_x q(s,t)-I, \\
P_0(s)&:=\nabla_x p(s,t),
\end{align*}
we observe
\begin{align}\label{eq xs}
\left\{
\begin{array}{l}
Q_0(s)=\int_t^s A(p(\tau,t))P_0(\tau) d\tau, \\
P_0(s)=-\int_t^s \nabla_x^2 V_\rho(\tau,q(\tau,t))Q_0(\tau) d\tau -\int_t^s \nabla_x^2 V_\rho(\tau,q(\tau,t)) d\tau.
\end{array}
\right.
\end{align}
Thus combining the two equations in (\ref{eq xs}), we learn
\begin{align*}
P_0(s)=B_t(P_0(\cdot))(s)+R_0(s),
\end{align*}
where
\begin{align*}
&B_t(P(\cdot))(s):=-\int_t^s \nabla_x^2 V_\rho(\tau,q(\tau,t))\left[\int_t^\tau A(p(\sigma,t))P(\sigma) d\sigma\right] d\tau, \\
&R_0(s):=-\int_t^s \nabla_x^2 V_\rho(\tau,q(\tau,t))d\tau.
\end{align*}
Let $\|M(\cdot)\|_0 :=\sup_{0\leq s\leq t}\langle s\rangle^{1+\varepsilon_1}|M(s)|$ for $M\in C(\left[0,t\right]; M_d(\mathbb{R}))$. Then (\ref{Vrho}) implies
\begin{align*}
|B_t(P(\cdot))(s)|&\leq \int_s^t C_2\rho^{\varepsilon_0}\langle \tau\rangle^{-2-\varepsilon_1}\int_\tau^t |P(\sigma)| d\sigma d\tau \\
&\leq C_2\rho^{\varepsilon_0}\|P\|_0 \int_s^\infty \langle \tau\rangle^{-2-\varepsilon_1}\int_\tau^\infty \langle \sigma \rangle^{-1-\varepsilon_1} d\sigma d\tau \\
&\leq C_2C^\prime\rho^{\varepsilon_0}\langle s \rangle^{-1-2\varepsilon_1}\|P\|_0, \\
|R_0(s)|&\leq \int_s^t C_2\rho^{\varepsilon_0}\langle \tau\rangle^{-2-\varepsilon_1} d\tau \leq C\rho^{\varepsilon_0}\langle s\rangle^{-1-\varepsilon_1}.
\end{align*}
If $\rho\leq (2C_2C^\prime)^{-\frac{1}{\varepsilon_0}}$, the operator norm $\| B_t \|_0$ of $B_t$ with respect to $\|\cdot\|_0$ is bounded by $\frac{1}{2}$. Hence we obtain
\begin{align}\label{px for 0}
\|P_0(\cdot)\|_0=\|(1-{B_t})^{-1}(R_0(\cdot))\|_0\leq \frac{1}{1-\|B_t\|_0}\|R_0(\cdot)\|_0 \leq 2C\rho^{\varepsilon_0},
\end{align}
which proves (\ref{ps}) for $\alpha=0$. The inequality (\ref{qs}) for $\alpha=0$ follows directly from (\ref{eq xs}) and (\ref{px for 0}).

Next we confirm the induction is valid. We fix $\alpha\in \mathbb{Z}_{+}^d \backslash \{0\}$ and assume that (\ref{qs}) and (\ref{ps}) hold for $\alpha^\prime$ with $|\alpha^\prime| < |\alpha|$.
Differentiating (\ref{eq xs}), we have
\begin{align}
\left\{
\begin{array}{l}
\partial_x^\alpha Q_0(s)=\int_t^s A(p(\tau,t))\partial_x^\alpha P_0(\tau) d\tau+R_{0,1}(s), \\
\partial_x^\alpha P_0(s)=-\int_t^s \nabla_x^2 V_\rho(\tau,q(\tau,t))\partial_x^\alpha Q_0(\tau) d\tau \\
\hspace{\fill}+R_{0,21}(s)+R_{0,22}(s),
\end{array}
\right.
\end{align}
where
\begin{align*}
R_{0,1}(s):=& \sum_{0\lneq\alpha^\prime\leq\alpha} \binom{\alpha}{\alpha^\prime} \int_t^s \partial_x^{\alpha^\prime}\left[A(p(\tau,t))\right]\partial_x^{\alpha-\alpha^\prime} P_0(\tau) d\tau , \\
R_{0,21}(s):=& -\sum_{0\lneq\alpha^\prime\leq\alpha}\binom{\alpha}{\alpha^\prime} \int_t^s  \partial_x^{\alpha^\prime}\left[\nabla_x^2V_\rho(\tau,q(\tau,t))\right]\partial_x^{\alpha-\alpha^\prime}Q_0(\tau) d\tau , \\
R_{0,22}(s):=& -\int_t^s \partial_x^\alpha\left[\nabla_x^2 V_\rho(\tau,q(\tau,t))\right] d\tau,
\end{align*}
and $\binom{\alpha}{\alpha^\prime}:=\prod_{j=1}^d \frac{\alpha_j !}{\alpha_j^\prime ! (\alpha_j-\alpha_j^\prime) !}$.
By (\ref{Vrho}) and assumptions of the induction, we have
\begin{align*}
|R_{0,1}(s)|&\leq C \rho^{\varepsilon_0}\langle s\rangle^{-1-\varepsilon_1}, \\
|R_{0,21}(s)|&\leq \int_s^t C\rho^{\varepsilon_0}\langle \tau\rangle^{-2-\varepsilon_1}\cdot C\rho^{\varepsilon_0}\langle \tau \rangle^{-\varepsilon_1} d\tau \leq C \rho^{\varepsilon_0}\langle s\rangle^{-1-2\varepsilon_1}, \\
|R_{0,22}(s)|&\leq  \int_s^t C\rho^{\varepsilon_0}\langle \tau\rangle^{-2-\varepsilon_1} d\tau \leq C \rho^{\varepsilon_0}\langle s\rangle^{-1-\varepsilon_1}.
\end{align*}
The similar argument as for $\alpha=0$ implies $\| \partial_x^\alpha P_0(\cdot) \|_0 \leq C_\alpha \rho^{\varepsilon_0}$ and (\ref{qs}).

\underline{Step 3}: Proof of (\ref{qxi}) and (\ref{pxi}). We use the induction with respect to $|\beta|$. First we consider the $\beta=0$ case. Similarly to Step 2, we have
\begin{align*}
\left\{
\begin{array}{l}
\nabla_\xi q(t,s)=\int_s^t A(p(\tau,s))\nabla_\xi p(\tau,s) d\tau , \\
\nabla_\xi p(t,s)=I-\int_s^t \nabla_x^2 V_\rho(\tau,q(\tau,s))\nabla_\xi q(\tau,s) d\tau,
\end{array}
\right.
\end{align*}
equivalently,
\begin{align}\label{eq xi0}
\left\{
\begin{array}{l}
Q^\prime(t)=\int_s^t A(p(\tau,s))P^\prime(\tau) d\tau -\int_s^t (A(p(\tau,s))-A(\xi))d\tau, \\
P^\prime(t)=-\int_s^t \nabla_x^2 V_\rho(\tau,q(\tau,s))Q^\prime(\tau) d\tau \\
\hspace{3.5cm}-\int_s^t (\tau-s)\nabla_x^2 V_\rho(\tau,q(\tau,s))A(\xi) d\tau,
\end{array}
\right.
\end{align}
where
\begin{align*}
Q^\prime(t)&:=\nabla_\xi q(t,s)-(t-s)A(\xi), \\
P^\prime(t)&:=\nabla_\xi p(t,s)-I.
\end{align*}
By (\ref{eq xi0}), we have
\begin{align*}
P^\prime(t)=B_s(P^\prime(\cdot))(t)+R^\prime(t),
\end{align*}
where
\begin{align*}
R^\prime(t):=-\int_s^t \nabla_x^2 V_\rho(\tau,q(\tau,s))\int_s^\tau A(p(\sigma,s)) d\sigma d\tau.
\end{align*}
Letting $\|M(\cdot)\|_1:=\sup_{t\geq s}|M(t)|$ for $M\in C(\left[s,\infty\right); M_d(\mathbb{R}))$, we have
\begin{align*}
|B_s(P(\cdot))(t)|&\leq \int_s^t C_2\rho^{\varepsilon_0}\langle \tau\rangle^{-2-\varepsilon_1}\int_s^\tau |P(\sigma)| d\sigma d\tau \\
&\leq C_2\rho^{\varepsilon_0}\|P\|_1\int_s^t \langle \tau\rangle^{-2-\varepsilon_1}(\tau-s) d\tau \\
&\leq C_2C^\prime\rho^{\varepsilon_1}\langle s \rangle^{-\varepsilon_1}\|P\|_1, \\
|R^\prime(t)|&\leq \int_s^t C\rho^{\varepsilon_1}\langle \tau \rangle^{-2-\varepsilon_1}(\tau-s) d\tau \leq C\rho^{\varepsilon_0}\langle s \rangle^{-\varepsilon_1}.
\end{align*}
Thus, if $\rho\leq (2C_2C^\prime)^{-\varepsilon_0}$, we obtain
\begin{align}\label{pxi for 0}
\|P^\prime(\cdot)\|_1=\|(1-B_s)^{-1}R^\prime(\cdot)\|_1 \leq \frac{1}{1-\|B_s\|_1}\|R^\prime(\cdot)\|_1 \leq 2C\rho^{\varepsilon_0}\langle s\rangle^{-\varepsilon_1}.
\end{align}
This proves (\ref{pxi}) for $\beta=0$. The inequality (\ref{qxi}) for $\beta=0$ follows from (\ref{p2}), (\ref{eq xi0}) and (\ref{pxi for 0}).

Next we prove the induction works.
Differentiating (\ref{eq xi0}), we have
\begin{align}
\left\{
\begin{array}{l}
\partial_\xi^\beta Q^\prime(t)=\int_s^t A(p(\tau,s))\partial_\xi^\beta P^\prime(\tau) d\tau +R^\prime_{11}(t)+R^\prime_{12}(t) , \\
\partial_\xi^\beta P^\prime(t)=-\int_s^t \nabla_x^2 V_\rho(\tau,q(\tau,s))\partial_\xi^\beta Q^\prime(\tau) d\tau +R^\prime_{21}(t)+R^\prime_{22}(t), \\
\end{array}
\right.
\end{align}
where
\begin{align*}
&R^\prime_{11}(t):= \sum_{0\lneq\beta^\prime\leq\beta} \binom{\beta}{\beta^\prime} \int_s^t \partial_\xi^{\beta^\prime} \left[A(p(\tau,s))\right]\partial_\xi^{\beta-\beta^\prime} P^\prime(\tau) d\tau , \\
&R^\prime_{12}(t):= \int_s^t \partial_\xi^\beta\left[A(p(\tau,s))-A(\xi) \right] d\tau , \\
&R^\prime_{21}(t):= -\sum_{0\lneq\beta^\prime\leq\beta}\binom{\beta}{\beta^\prime} \int_s^t \partial_\xi^{\beta^\prime}\left[\nabla_x^2V_\rho(\tau,q(\tau,s))\right]\partial_\xi^{\beta-\beta^\prime} Q^\prime(\tau) d\tau , \\
&R^\prime_{22}(t):= -\int_s^t (\tau-s)\partial_\xi^\beta\left[\nabla_x^2 V_\rho(\tau,q(\tau,s))A(\xi)\right] d\tau .
\end{align*}
Thus we have
\begin{align*}
\partial_\xi^\beta P^\prime(t)&=B_s(\partial_\xi^\beta P^\prime(\cdot))(t)-\int_s^t \nabla_x^2 V_\rho(\tau,q(\tau,s))(R^\prime_{11}(\tau)+R^\prime_{12}(\tau)) d\tau \\
&\hspace{7cm}+R^\prime_{21}(t)+R^\prime_{22}(t).
\end{align*}
If (\ref{qxi}) and (\ref{pxi}) are true for $\beta^\prime$ with $|\beta^\prime|<|\beta|$, we learn
\begin{align*}
|R^\prime_{11}(t)|&\leq C \rho^{\varepsilon_0}\langle s\rangle^{-\varepsilon_1}|t-s| , \\
|R^\prime_{12}(t)|&\leq C \sup_{|\beta^\prime|\leq|\beta|} \int_s^t |\partial_\xi^{\beta^\prime}\left[p(\tau,s)-\xi \right]| d\tau \leq C \rho^{\varepsilon_0}\langle s\rangle^{-\varepsilon_1}|t-s| , \\
|R^\prime_{21}(t)|&\leq \int_s^t C\rho^{\varepsilon_0}\langle \tau\rangle^{-2-\varepsilon_1}\cdot C\rho^{\varepsilon_0}\langle s\rangle^{-\varepsilon_1}|\tau-s| d\tau \leq C \rho^{2\varepsilon_0}\langle s\rangle^{-2\varepsilon_1} , \\
|R^\prime_{22}(t)|&\leq  \int_s^t C\rho^{\varepsilon_0}\langle \tau\rangle^{-2-\varepsilon_1}|\tau-s| d\tau \leq C \rho^{\varepsilon_0}\langle s\rangle^{-\varepsilon_1}.
\end{align*}
Using the similar argument as for $\beta=0$, we obtain (\ref{qxi}) and (\ref{pxi}) for any $\beta$.

\underline{Step 4}: Proof of (\ref{qx}) and (\ref{px}). We use the induction with respect to $|\alpha|+|\beta|$. In the $\alpha=\beta=0$ case, differentiation in $x$ implies
\begin{align*}
\left\{
\begin{array}{l}
\nabla_x q(t,s)=I+\int_s^t A(p(\tau,s))\nabla_x p(\tau,s) d\tau , \\
\nabla_x p(t,s)=-\int_s^t \nabla_x^2 V_\rho(\tau,q(\tau,s))\nabla_x q(\tau,s) d\tau.
\end{array}
\right.
\end{align*}
Letting 
\begin{align*}
Q(t)&:=\nabla_x q(t,s)-I, \\
P(t)&:=\nabla_x p(t,s),
\end{align*}
we observe
\begin{align}
\left\{
\begin{array}{l}\label{eq x0}
Q(t)=\int_s^t A(p(\tau,s))P(\tau) d\tau , \\
P(t)=-\int_s^t \nabla_x^2 V_\rho(\tau,q(\tau,s))Q(\tau) d\tau -\int_s^t \nabla_x^2 V_\rho(\tau,q(\tau,s)) d\tau.
\end{array}
\right.
\end{align}
This implies
\begin{align*}
P(t)=B_s(P(\cdot))(t)+R(t),
\end{align*}
where
\begin{align*}
R(t):=-\int_s^t \nabla_x^2 V_\rho(\tau,q(\tau,s))d\tau.
\end{align*}
Since
\begin{align*}
|R(t)|\leq \int_s^t C_2\rho^{\varepsilon_0}\langle \tau\rangle^{-2-\varepsilon_1} d\tau \leq C\rho^{\varepsilon_0}\langle s\rangle^{-1-\varepsilon_1},
\end{align*}
we have
\begin{align*}
\|P(\cdot)\|_1=\|(1-B_s)^{-1}R\|_1 \leq 2C\rho^{\varepsilon_0}\langle s \rangle^{-1-\varepsilon_1},
\end{align*}
which proves (\ref{px}) for $\alpha=\beta=0$. The inequality (\ref{qx}) follows from (\ref{px}) and (\ref{eq x0}).

We prove the induction with respect to $|\alpha|+|\beta|$ works. By (\ref{eq x0}), we have
\begin{align}
\left\{
\begin{array}{l}
\partial_x^\alpha\partial_\xi^\beta Q(t)=\int_s^t A(p(\tau,s))\partial_x^\alpha\partial_\xi^\beta P(\tau) d\tau +R_1(t) , \\
\partial_x^\alpha\partial_\xi^\beta P(t)=-\int_s^t \nabla_x^2 V_\rho(\tau,q(\tau,s))\partial_x^\alpha\partial_\xi^\beta Q(\tau) d\tau\hspace{0.5cm} \\
\hspace{\fill}+R_{21}(t)+R_{22}(t),
\end{array}
\right.
\end{align}
where
\begin{align*}
R_{1}(t)&:= \sum_{\substack{\alpha^\prime\leq\alpha,\beta^\prime\leq\beta, \\ |\alpha^\prime+\beta^\prime|\geq1}} \binom{\alpha}{\alpha^\prime}\binom{\beta}{\beta^\prime} \int_s^t \partial_x^{\alpha^\prime}\partial_\xi^{\beta^\prime}\left[A(p(\tau,s))\right]
\partial_x^{\alpha-\alpha^\prime}\partial_\xi^{\beta-\beta^\prime} P(\tau) d\tau , \\
R_{21}(t)& \\
:= -&\sum_{\substack{\alpha^\prime\leq\alpha,\beta^\prime\leq\beta, \\ |\alpha^\prime+\beta^\prime|\geq1}}
\binom{\alpha}{\alpha^\prime}\binom{\beta}{\beta^\prime}
\int_s^t  \partial_x^{\alpha^\prime}\partial_\xi^{\beta^\prime}\left[\nabla_x^2V_\rho(\tau,q(\tau,s))\right]
\partial_x^{\alpha-\alpha^\prime}\partial_\xi^{\alpha-\beta^\prime} Q(\tau) d\tau , \\
R_{22}(t)&:= -\int_s^t \partial_x^\alpha\partial_\xi^\beta\left[\nabla_x^2 V_\rho(\tau,q(\tau,s))\right] d\tau .
\end{align*}
Thus we learn
\begin{align*}
\partial_x^\alpha\partial_\xi^\beta P(t)&=B_s(\partial_x^\alpha\partial_\xi^\beta P(\cdot))(t)-\int_s^t \nabla_x^2 V_\rho(\tau,q(\tau,s))R_1(\tau) d\tau \\
&\hspace{6cm}+R_{21}(t)+R_{22}(t).
\end{align*}
By (\ref{qxi}),(\ref{pxi}) and assumptions of the induction, we have
\begin{align*}
|R_1(t)|&\leq C \rho^{\varepsilon_0}\langle s\rangle^{-1-\varepsilon_1}|t-s| , \\
|R_{21}(t)|&\leq \int_s^t C\rho^{\varepsilon_0}\langle \tau\rangle^{-2-\varepsilon_1}\cdot C\rho^{\varepsilon_0}\langle s\rangle^{-1-\varepsilon_1}|\tau-s| d\tau \leq C \rho^{2\varepsilon_0}\langle s\rangle^{-1-2\varepsilon_1} , \\
|R_{22}(t)|&\leq  \int_s^t C\rho^{\varepsilon_0}\langle \tau\rangle^{-2-\varepsilon_1} d\tau \leq C \rho^{\varepsilon_0}\langle s\rangle^{-1-\varepsilon_1} .
\end{align*}
Similarly to the argument for $\alpha=\beta=0$, we obtain (\ref{qx}) and (\ref{px}) for any $\alpha$ and $\beta$.

\underline{Step 5}: Proof of (\ref{qesti}). By (\ref{canq}) and (\ref{canp}), we have
\begin{align*}
q(t,s;x,\xi)&=x+\int_s^t v(p(\tau,s)) d\tau \\
&=x+\int_s^t v\left(p(t,s)+\int_\tau^t\nabla_x V_\rho(\sigma,q(\sigma,s)) d\sigma\right) d\tau .
\end{align*}
Thus
\begin{align*}
&q(t,s;x,\xi)-x-(t-s)v(p(t,s)) \\
&=\int_s^t \left[v\left(p(t,s)+\int_\tau^t \nabla_x V_\rho(\sigma,q(\sigma,s))d\sigma\right)-v(p(t,s))\right] d\tau .
\end{align*}
This equality and (\ref{qx})-(\ref{pxi}) imply (\ref{qesti}).
\end{proof}

Similarly to Proposition 2.2 in \cite{IsoKita}, we observe that, if $\rho$ is small enough, the maps
\begin{align*}
y &\mapsto q(s,t;y,\xi), \\
\eta &\mapsto p(t,s;x,\eta)
\end{align*}
have the corresponding inverses.

\begin{prop} \label{y eta estimates}
Fix $\rho>0$ so that $C_0\rho^{\varepsilon_0}<\frac{1}{2}$ holds, where $C_0$ is the constant in Proposition \ref{general estimates}. Then, for $x,\xi\in\mathbb{R}^d$ and $0\leq\pm s \leq\pm t $, there exist $y(s,t)=y(s,t;x,\xi)\in\mathbb{R}^d$ and $\eta(t,s)=\eta(t,s;x,\xi)\in\mathbb{R}^d$ such that 
\begin{numcases}{}
q(s,t;y(s,t;x,\xi),\xi)=x, \label{y} \\
p(t,s;x,\eta(t,s;x,\xi))=\xi, \label{eta}
\end{numcases}
and
\begin{numcases}{}
q(t,s;x,\eta(t,s;x,\xi))=y(s,t;x,\xi), \label{q eta=y} \\
p(s,t;y(s,t;x,\xi),\xi)=\eta(t,s;x,\xi). \label{p y=eta}
\end{numcases}
Furthermore, for any $x,\xi\in\mathbb{R}^d$, $0\leq\pm s \leq\pm t $ and multi-indices $\alpha$ and $\beta$,
\begin{align}
&|\partial_x^\alpha\left[\nabla_x y(s,t;x,\xi)-I \right]|\leq C^\prime_{\alpha}\rho^{\varepsilon_0}\langle s\rangle^{-\varepsilon_1}, \label{x y} \\ 
&|\partial_x^\alpha\partial_\xi^\beta\nabla_x \eta(t,s;x,\xi)|\leq C^\prime_{\alpha\beta}\rho^{\varepsilon_0}\langle s\rangle^{-1-\varepsilon_1}, \label{x eta} \\
&|\partial_\xi^\beta \left[\eta(t,s;x,\xi)-\xi\right]|\leq C^\prime_{\beta}\rho^{\varepsilon_0}\langle s\rangle^{-\varepsilon_1}, \label{xi eta} \\
\label{xi y} &|\partial_\xi^\beta\left[y(s,t;x,\xi)-x-(t-s)v(\xi)\right]| \\
&\leq C^\prime_{\beta}\rho^{\varepsilon_0}\min\{|t-s|\langle s\rangle^{-\varepsilon_1},\langle t\rangle^{1-\varepsilon_1}\}. \nonumber 
\end{align}
\end{prop}

\begin{proof}
\underline{Step 1}. By $|\nabla_x q(s,t;x,\xi)-I|<\frac{1}{2}$, $|\nabla_\xi p(t,s;x,\xi)-I|<\frac{1}{2}$ and Schwartz's global inversion theorem (\cite{D-G}, Proposition A.7.1), we have the existence and uniqueness of $y(s,t;x,\xi)$ and $\eta(t,s;x,\xi)$ satisfying (\ref{y}) and (\ref{eta}).
The equalities (\ref{q eta=y}) and (\ref{p y=eta}) are shown by (\ref{y}) and (\ref{eta}).

\underline{Step 2}: Proof of (\ref{x y}). Differentiation of (\ref{y}) in $x$ implies
\begin{align} \label{qx yx}
\nabla_x q(s,t;y(s,t),\xi) \nabla_x y(s,t)=I.
\end{align}
We have by (\ref{qs})
\begin{align*}
|\nabla_x y(s,t)-I|&=|(\nabla_x q(s,t;y(s,t),\xi))^{-1}-I| \\
&\leq C|\nabla_x q(s,t;y(s,t),\xi)-I| \\
&\leq C\rho^{\varepsilon_0}\langle s\rangle^{-\varepsilon_1}.
\end{align*}

Differentiating (\ref{qx yx}), we have for $\alpha\neq0$
\begin{align*}
&\nabla_x q(s,t;y(s,t),\xi) \partial_x^\alpha\nabla_x y(s,t) \\
&=-\sum_{0\lneq\alpha^\prime\leq\alpha} \binom{\alpha}{\alpha^\prime} \partial_x^{\alpha^\prime}\left[\nabla_x q(s,t;y(s,t),\xi)\right] \partial_x^{\alpha-\alpha^\prime}\nabla_x y(s,t).
\end{align*}
Using (\ref{qs}) and the induction with respect to $|\alpha|$, we observe that the RHS of the above equality is bounded by $C\rho^{\varepsilon_0}\langle s\rangle^{-\varepsilon_1}$.
Thus we have $|\partial_x^\alpha\nabla_x y(s,t)|\leq C^\prime_{\alpha}\rho^{\varepsilon_0}\langle s\rangle^{-\varepsilon_1}$.

\underline{Step 3}: Proof of (\ref{xi eta}). By (\ref{p y=eta}), we observe for $\beta=0$
\begin{align*}
|\eta(t,s)-\xi|&=|p(s,t;y(s,t),\xi)-\xi| \\
&=\left|\int_s^t \nabla_x V_\rho(\tau,q(\tau,t;y(s,t),\xi)) d\tau\right| \\
&\leq C\rho^{\varepsilon_0}\langle s\rangle^{-\varepsilon_1}.
\end{align*}

In the case of $|\beta|=1$, we have by differentiation of (\ref{eta}) in $\xi$
\begin{align*}
\nabla_\xi p(t,s;x,\eta(t,s)) \nabla_\xi \eta(t,s)=I.
\end{align*}
Similarly to Step 2, we obtain by (\ref{pxi})
\begin{align*}
|\nabla_\xi \eta(t,s)-I|&\leq C|\nabla_\xi p(t,s;x,\eta(t,s))-I| \\
&\leq C\rho^{\varepsilon_0}\langle s\rangle^{-\varepsilon_1}.
\end{align*}

In the other cases, we learn by (\ref{eta})
\begin{align*}
&\nabla_\xi p(t,s;x,\eta(t,s)) \partial_\xi^\beta\nabla_\xi \eta(t,s) \\
&=-\sum_{0\lneq\beta^\prime \leq \beta} \binom{\beta}{\beta^\prime} \partial_\xi^{\beta^\prime}[\nabla_\xi p(t,s;x,\eta(t,s))] \partial_\xi^{\beta-\beta^\prime}\nabla_\xi \eta(t,s), \quad \beta \neq 0.
\end{align*}
The induction with respect to $|\beta|$ and (\ref{pxi}) imply each term in the RHS is bounded by $C\rho^{\varepsilon_0}\langle s\rangle^{-\varepsilon_1}$. Thus (\ref{xi eta}) holds for any $\beta$.

\underline{Step 4}: Proof of (\ref{x eta}). Differentiating (\ref{eta}) in $x$, we have
\begin{align*}
\nabla_x p(t,s;x,\eta(t,s))+\nabla_\xi p(t,s;x,\eta(t,s)) \nabla_x \eta(t,s)=0.
\end{align*}
This equality and (\ref{px}) imply
\begin{align*}
|\nabla_x \eta(t,s)|&=|(\nabla_\xi p(t,s;x,\eta(t,s)))^{-1}\nabla_x p(t,s;x,\eta(t,s))| \\
&\leq C|\nabla_x p(t,s;x,\eta(t,s))| \\
&\leq C\rho^{\varepsilon_0}\langle s\rangle^{-1-\varepsilon_1},
\end{align*}
which proves (\ref{x eta}) for $\alpha=\beta=0$.
If $\alpha+\beta\neq0$, we have
\begin{align*}
\nabla_\xi &p(t,s;x,\eta(t,s)) \partial_x^\alpha\partial_\xi^\beta\nabla_x \eta(t,s) \\
=&-\partial_x^\alpha\partial_\xi^\beta\left[\nabla_x p(t,s;x,\eta(t,s))\right] \\
&-\sum_{\substack{\alpha^\prime\leq\alpha,\beta^\prime\leq\beta, \\ |\alpha^\prime+\beta^\prime|\geq1}} \binom{\alpha}{\alpha^\prime}\binom{\beta}{\beta^\prime} \partial_x^{\alpha^\prime}\partial_\xi^{\beta^\prime}[\nabla_\xi p(t,s;x,\eta(t,s))] \partial_x^{\alpha-\alpha^\prime}\partial_\xi^{\beta-\beta^\prime}\nabla_x \eta(t,s).
\end{align*}
Thus (\ref{x eta}) is proved by (\ref{xi eta}), (\ref{px}), (\ref{pxi}) and the induction with respect to $|\alpha|+|\beta|$.

\underline{Step 5}: Proof of (\ref{xi y}). Similarly to the proof of (\ref{qesti}) in Proposition \ref{general estimates}, we have
\begin{align*}
&y(s,t)-x-(t-s)v(\xi) \\
&=q(t,s;x,\eta(t,s))-x-(t-s)v(p(t,s;x,\eta(t,s))) \\
&=\int_s^t \left[v\left(\xi+\int_\tau^t \nabla_x V_\rho(\sigma,q(\sigma,s;x,\eta(t,s)))d\sigma\right)-v(\xi)\right] d\tau.
\end{align*}
Using this equality, (\ref{qxi}) and (\ref{xi eta}), we obtain (\ref{xi y}).
\end{proof}

We define
\begin{align*}
\phi(t;x,\xi):=u(t;x,\eta(t,0;x,\xi)),
\end{align*}
where
\begin{align*}
u(t;x,\eta):=x\cdot\eta+\int_0^t\{h_\rho-x\cdot\nabla_x h_\rho\}(\tau,q(\tau,0;x,\eta),p(\tau,0;x,\eta))d\tau .
\end{align*}
Then a direct calculus implies that $\phi$ satisfies the Hamilton-Jacobi equation
\begin{align} \label{Hamilton-Jacobi}
\left\{
\begin{array}{l}
\partial_t\phi(t;x,\xi)=h_\rho(t,\nabla_\xi\phi(t;x,\xi),\xi), \\
\phi(0;x,\xi)=x\cdot\xi,
\end{array}
\right.
\end{align}
and the relation between $\phi$ and the functions $y$ and $\eta$ in Proposition \ref{y eta estimates}:
\begin{align} \label{phi relations}
\left\{
\begin{array}{l}
\nabla_x\phi(t;x,\xi)=\eta(t,0;x,\xi), \\
\nabla_\xi\phi(t;x,\xi)=y(0,t;x,\xi).
\end{array}
\right.
\end{align}

\begin{rem}
The relation (\ref{phi relations}) and Proposition \ref{y eta estimates} imply the estimate
\begin{align}\label{x y improved}
|\partial_x^\alpha\partial_\xi^\beta\left[\nabla_x y(s,t;x,\xi)-I \right]|\leq C^\prime_{|\alpha|+|\beta|}\rho^{\varepsilon_0}\langle s\rangle^{-\varepsilon_1}
\end{align}
holds for $|\beta|\geq1$. Hence (\ref{x y}) is extended to (\ref{x y improved}) for any $\alpha$ and $\beta$.
\end{rem}

Now, we construct outgoing and incoming solutions of the eikonal equation (\ref{eikonal0}).

\begin{lem}
The limits
\begin{align} \label{phi definition}
\phi_\pm(x,\xi):=\lim_{t\to\pm\infty} (\phi(t;x,\xi)-\phi(t;0,\xi))
\end{align}
exist, are smooth in $\mathbb{R}^{2d}$ and
\begin{align}\label{periodic}
\phi_\pm(x,\xi+2\pi m)=\phi_\pm(x,\xi)+2\pi x\cdot m, \quad x,\xi\in\mathbb{R}^d, \ m\in\mathbb{Z}^d.
\end{align}
\end{lem}

\begin{proof}
We define
\begin{align*}
R(t,x,\xi):=\phi(t;x,\xi)-\phi(t;0,\xi).
\end{align*}
Then we have
\begin{align*}
\nabla_x R(t,x,\xi)&=\eta(t,0;x,\xi)=p(0,t;y(0,t;x,\xi),\xi) \\
&=\xi+\int_0^t(\nabla_x V_\rho)(\tau,q(\tau,t;y(0,t;x,\xi),\xi))d\tau \nonumber \\
&=\xi+\int_0^t(\nabla_x V_\rho)(\tau,q(\tau,0;x,\eta(t,0;x,\xi)))d\tau \nonumber .
\end{align*}
Since
\begin{align*}
|\partial_x^\alpha\partial_\xi^\beta [(\nabla_x V_\rho)(\tau,q(\tau,0;x,\eta(t,0;x,\xi)))]|\leq C_{\alpha\beta}\langle\tau\rangle^{-1-\varepsilon_1},
\end{align*}
$\nabla_x R(t,x,\xi)$ converges to a smooth function uniformly in $(x,\xi)\in\mathbb{R}^{2d}$. Thus
\begin{align} \label{R representation}
\partial_\xi^\beta R(t,x,\xi)=x\cdot\int_0^1 \nabla_x \partial_\xi^\beta R(t,\theta x,\xi) d\theta
\end{align}
converges locally uniformly in $\mathbb{R}^{2d}$. This implies the smoothness of $\phi_\pm$.

It is easy to see (\ref{periodic}) if we remark
\begin{align*}
\eta(t,0;x,\xi+2\pi m)&=\eta(t,0;x,\xi)+2\pi m , \\
q(t,0;x,\xi+2\pi m)&=q(t,0;x,\xi)
\end{align*}
for $x,\xi\in\mathbb{R}^d$, $t\in\mathbb{R}$ and $m\in\mathbb{Z}^d$.
\end{proof}

Next we consider properties of $\phi_\pm$ in the ``outgoing'' and ``incoming'' regions.
We prepare improved estimates of Proposition \ref{general estimates} for an orbit which is outgoing or incoming. 

\begin{lem} \label{outgoing estimates}
Let $(q,p)(t)=(q,p)(t,0;x,\xi)$ be an orbit satisfying (\ref{canq}) and (\ref{canp}). Suppose
\begin{align*} 
|q(\tau)|\geq b|\tau|+d,\quad \pm \tau\geq 0
\end{align*}
for some $b>0$ and $d\geq0$.
Then there exist $l_{\alpha\beta}, l_\beta \geq2$ such that for $\pm t\geq 0$ and $\alpha$, $\beta\in\mathbb{N}_{\geq0}^d$,
\begin{align}
&|p(t)-\xi|\leq C b^{-1} \langle d\rangle^{-\varepsilon}, \label{p2 o}\\
&|\partial_x^\alpha\partial_\xi^\beta\left[\nabla_x q(t)-I \right]|\leq C_{\alpha\beta} b^{-l_{\alpha\beta}} \langle d\rangle^{-1-|\alpha|-\varepsilon}|t|,\label{qx o} \\ 
&|\partial_x^\alpha\partial_\xi^\beta \nabla_x p(t)|\leq C_{\alpha\beta} b^{-l_{\alpha\beta}} \langle d\rangle^{-1-|\alpha|-\varepsilon}, \label{px o}\\
&|\partial_\xi^\beta\left[\nabla_\xi q(t)-tA(\xi) \right]|\leq C_{\beta} b^{-l_{\beta}} \langle d\rangle^{-\varepsilon}|t|,\label{qxi o} \\ 
&|\partial_\xi^\beta\left[\nabla_\xi p(t)-I\right]|\leq C_{\beta} b^{-l_{\beta}} \langle d\rangle^{-\varepsilon}.\label{pxi o}
\end{align}
\end{lem}

\begin{proof}
We calculate similarly to Proposition \ref{general estimates}, whereas we use the following estimate instead:
\begin{align*}
|\partial_x^\alpha V_\rho(t,q(t))| \leq C_\alpha \langle q(t) \rangle^{-|\alpha|-\varepsilon}
\leq C_\alpha \langle b|t|+d \rangle^{-|\alpha|-\varepsilon}.
\end{align*}

\end{proof}

The next lemma gives improved estimates of Proposition \ref{y eta estimates} for outgoing or incoming orbits.

\begin{lem}\label{phi estimate}
Let $b,d\geq0$, $b\neq0$ and $x$, $\xi\in\mathbb{R}^d$ satisfy
\begin{align*}
|q(\tau,0;x,\eta(t,0;x,\xi))| \geq b|\tau|+d , \quad 0 \leq \pm\tau \leq \pm t
\end{align*}
for any $\pm t \geq 0$. Then there exist $l_{\alpha\beta}^\prime, l_\beta^\prime \geq2$ such that, for $\pm t \geq 0$,
\begin{align}
&|\partial_x^\alpha\partial_\xi^\beta\left[\nabla_x \eta(t,0;x,\xi)\right]|\leq C_{\alpha\beta} b^{-l^\prime_{\alpha\beta}} \langle d\rangle^{-1-|\alpha|-\varepsilon}, \label{x eta o} \\
&|\partial_\xi^\beta \left[\eta(t,0;x,\xi)-\xi\right]|\leq C_{\beta} b^{-l^\prime_{\beta}} \langle d\rangle^{-\varepsilon}. \label{xi eta o} 
\end{align}
\end{lem}

\begin{proof}
The proofs are similar to those of (\ref{x eta}) and (\ref{xi eta}) if we use
\begin{align*}
|\partial_x^\alpha V_\rho(\tau,q(\tau,0;x,\eta(t,0;x,\xi)))|
\leq C_\alpha \langle b|\tau|+d \rangle^{-|\alpha|-\varepsilon}, \quad 0 \leq \pm\tau \leq \pm t.
\end{align*}

\end{proof}

Using the above two lemmas, we have the estimate of $\phi_\pm(x,\xi)-x\cdot\xi$ on the outgoing and incoming region, respectively.
See Proposition 2.4 in \cite{IsoKita} for the case of Schr\"odinger operators.

\begin{prop}
\begin{align} \label{phi+- estimate}
|\partial_x^\alpha \partial_\xi^\beta[\phi_\pm(x,\xi)-x\cdot\xi]|
\leq C_{\alpha\beta} |v(\xi)|^{-l_{\alpha\beta}} \langle x\rangle^{1-|\alpha|-\varepsilon}
\end{align}
on $\{(x,\xi) \mid |x|^{\varepsilon_1} |v(\xi)|^{1-\varepsilon_1}\geq C_{\varepsilon_1}, \pm \cos(x,v(\xi))\geq 0 \}$, respectively.
\end{prop}

\begin{proof}
On $\{(x,\xi) \mid x, v(\xi)\neq0,\ \pm\cos(x,v(\xi))\geq0\}$, (\ref{p1}), (\ref{p2}) and (\ref{qesti}) imply for $0\leq\pm \tau \leq \pm t$,
\begin{align*}
|q(\tau,0;x,\eta(t,0;x,\xi))|\geq& |x+\tau v(p(\tau,0;x,\eta(t,0;x,\xi)))|-C_0\langle \tau \rangle^{1-\varepsilon_1} \\
=& |x+\tau v(p(\tau,t;y(0,t;x,\xi),\xi))|-C_0\langle \tau \rangle^{1-\varepsilon_1} \\
\geq& |x+\tau v(\xi)|-C \langle \tau \rangle^{1-\varepsilon_1}-C_0\langle \tau \rangle^{1-\varepsilon_1} \\
\geq& \frac{1}{\sqrt{2}}(|x|+|\tau v(\xi)|) -C\langle \tau \rangle^{1-\varepsilon_1}.
\end{align*}
If we remark
\begin{align*}
|x|+|\tau v(\xi)| \geq& \left(\frac{1}{\varepsilon_1}|x|\right)^{\varepsilon_1}\left(\frac{1}{1-\varepsilon_1}|\tau v(\xi)|\right)^{1-\varepsilon_1} 
=\frac{|x|^{\varepsilon_1}|v(\xi)|^{1-\varepsilon_1}}{\varepsilon_1^{\varepsilon_1} (1-\varepsilon_1)^{1-\varepsilon_1}}|\tau|^{1-\varepsilon_1},
\end{align*}
we learn for $|x|^{\varepsilon_1} |v(\xi)|^{1-\varepsilon_1}\geq C_{\varepsilon_1}$
\begin{align}\label{lower bounds}
|q(\tau,0;x,\eta(t,0;x,\xi))|\geq \frac{1}{2}(|x|+|\tau v(\xi)|) , \quad 0\leq \pm\tau\leq \pm t.
\end{align}
Hence the proposition is proved by (\ref{lower bounds}), (\ref{phi relations}), (\ref{phi definition}), (\ref{R representation}) and Lemma \ref{phi estimate}.
\end{proof}

The following proposition says $\phi_\pm$ is a solution to the eikonal equation (\ref{eikonal0}).

\begin{prop}\label{eikonal prop}
For any $a>0$, there exists $R_a>1$ such that $\phi_\pm$ satisfies the eikonal equation
\begin{align}\label{eikonal}
h(x,\nabla_x\phi_\pm(x,\xi))=h_0(\xi)
\end{align}
on the outgoing (or incoming) region
\begin{align*}
\{ (x,\xi) \mid |x|\geq R_a,\ |v(\xi)|\geq a, \pm\cos(x,v(\xi))\geq 0\},
\end{align*}
respectively.
\end{prop}

\begin{proof}
By (\ref{phi relations}) and (\ref{phi definition}), we have
\begin{align*}
\nabla_x \phi_\pm(x,\xi)=\lim_{t\to\pm\infty} \eta(t,0;x,\xi)=\lim_{t\to\pm\infty} p(0,t;y(0,t;x,\xi),\xi) .
\end{align*}
If $|x|\geq 2\rho^{-1}$, then we have by the definition of $V_\rho$
\begin{align}\label{energy1}
h(x,\nabla_x\phi_\pm(x,\xi))=\lim_{t\to\pm\infty} h_\rho(0,x,p(0,t;y(0,t;x,\xi),\xi)) .
\end{align}
Now we claim
\begin{align*}
E(\tau)&:=h_\rho(\tau,q(\tau,t;y(0,t;x,\xi),\xi),p(\tau,t;y(0,t;x,\xi),\xi)) \\
&=h_\rho(\tau,q(\tau,0;x,\eta(t,0;x,\xi)),p(\tau,0;x,\eta(t,0;x,\xi)))
\end{align*}
is a constant for $0\leq\pm \tau\leq\pm t$. A direct calculus implies
\begin{align*}
\frac{dE}{d\tau}(\tau)&=\partial_t h_\rho(\tau,q(\tau,0;x,\eta(t,0;x,\xi)),p(\tau,0;x,\eta(t,0;x,\xi))) \\
&=\partial_t V_\rho(\tau,q(\tau,0;x,\eta(t,0;x,\xi))).
\end{align*}
We note (\ref{lower bounds}) holds on $\{ (x,\xi) \mid |x|\geq R_a,\ |v(\xi)|\geq a,\ \pm\cos(x,v(\xi))\geq0\}$ for $R_a$ large enough, and hence
\begin{align*}
|q(\tau,0;x,\eta(t,0;x,\xi))|&\geq \frac{1}{2}(R_a+a|\tau|) \\
&\geq2\max\{\rho^{-1},\frac{\langle \tau\rangle}{\langle\log\langle \tau \rangle\rangle}\}, \quad 0\leq\pm\tau\leq\pm t.
\end{align*}
We also note $\partial_t V_\rho(t,x)=0$ if $|x|\geq2\max\{\rho^{-1},\frac{\langle t\rangle}{\langle\log\langle t \rangle\rangle}\}$.
Thus we have $\frac{dE}{d\tau}(\tau)=0$ if $ 0\leq \pm\tau \leq \pm t$, in particular,
\begin{align}\label{energy2}
h_\rho(0,x,p(0,t;y(0,t;x,\xi),\xi))&=E(0)=E(t) \\
&=h_\rho(t,y(0,t;x,\xi),\xi). \nonumber
\end{align}
Hence, (\ref{energy1}) and (\ref{energy2}) imply
\begin{align*}
h(x,\nabla_x\phi_\pm(x,\xi))=\lim_{t\to\pm\infty} h_\rho(t,y(0,t;x,\xi),\xi)=h_0(\xi).
\end{align*}
\end{proof}

\begin{proof}[Proof of Proposition \ref{mainp}]
Let $\varphi\in C^\infty(\mathbb{R}^d\times(\mathbb{R}^d\backslash v^{-1}(0)))$ be defined by
\begin{align} \label{def of phase}
\varphi(x,\xi)=(\phi_+(x,\xi)&-x\cdot \xi)\chi_+(x,\xi) \\
&+(\phi_-(x,\xi)-x\cdot\xi)\chi_-(x,\xi)+x\cdot\xi , \nonumber
\end{align}
where
\begin{align}
\chi_\pm(x,\xi)=\chi\left(\mu|v(\xi)|^\ell x\right)\psi_\pm(\cos(x,v(\xi)))
\end{align}
and $\psi_\pm\in C^\infty([-1,1]; [0,1])$ satisfy
\begin{align*}
\psi_\pm(\sigma)=
\begin{cases}
1, \quad \pm\sigma\geq\frac{1}{2},\\
0, \quad \pm\sigma\leq 0.
\end{cases}
\end{align*}
If $\mu$ and $\ell$ are fixed so that $\mu$ is sufficiently small and that $\ell$ is sufficiently large, then $\varphi$ satisfies (\ref{shuuki}), (\ref{phase estimate}) and (\ref{phase estimate2}).

Finally we prove (\ref{s_a}).
Let $s_a$ be defined by (\ref{s_a def}). We decompose $s_a$ by
\begin{align}\label{s decomposition}
s_a(x,\xi)=s_a^1(x,\xi)+s_a^2(x,\xi), 
\end{align}
where
\begin{align*}
s_a^1(x,\xi)&=\sum_{z\in\mathbb{Z}^d}f\left[z\right]e^{i(\varphi_a(x-z,\xi)-\varphi_a(x,\xi))}-h_0(\nabla_x\varphi_a(x,\xi)) , \\
s_a^2(x,\xi)&=h(x,\nabla_x\varphi_a(x,\xi))-h_0(\xi) .
\end{align*}
For $s_a^2$, (\ref{eikonal}) and Assumption \ref{ass2} imply for $|x|\geq R_a$ and $\beta$,
\begin{align}\label{s2 estimate}
\partial_\xi^\beta s_a^2(x,\xi)=
\begin{cases}
0,\  |\cos(x,v(\xi))|\geq\frac{1}{2}, \\
\mathcal{O}(\langle x\rangle^{-\varepsilon}),\  |\cos(x,v(\xi))|\leq\frac{1}{2}.
\end{cases}
\end{align}
For $s_a^1$, we have
\begin{align*}
s_a^1(x,\xi)&=\sum_{z\in\mathbb{Z}^d}f\left[z\right]\left(e^{i(\varphi_a(x-z,\xi)-\varphi_a(x,\xi))}-e^{-iz\cdot\nabla_x\varphi_a(x,\xi)}\right) \\
&=\sum_{z\in\mathbb{Z}^d}f\left[z\right]e^{-iz\cdot\nabla_x\varphi_a(x,\xi)} \left(e^{i\Phi_a(x,\xi,z)}-1\right) ,
\end{align*}
where
\begin{align*}
\Phi_a(x,\xi,z)&=\varphi_a(x-z,\xi)-\varphi_a(x,\xi)+z\cdot\nabla_x\varphi_a(x,\xi) \\
&=z\cdot \left(\int_0^1 \theta_1 \int_0^1 \nabla_x^2\varphi_a(x-\theta_1\theta_2 z,\xi) d\theta_2 d\theta_1\right)z .
\end{align*}
By (\ref{phase estimate}), we observe
\begin{align*}
|\partial_\xi^\beta[e^{-iz\cdot\nabla_x\varphi_a(x,\xi)}]|&\leq C_\beta \langle z\rangle^{|\beta|}
\end{align*}
and
\begin{align*}
|\partial_\xi^\beta \Phi_a(x,\xi,z)|\leq& C_\beta |z|^2 \int_0^1 \theta_1 \int_0^1 \langle x-\theta_1\theta_2 z \rangle^{-1-\varepsilon} d\theta_2 d\theta_1 \\
\leq& C_\beta \langle x\rangle^{-1-\varepsilon} \langle z\rangle^{3+\varepsilon}.
\end{align*}
Thus we obtain
\begin{align}\label{s1 estimate}
|\partial_\xi^\beta s_a^1(x,\xi)|\leq C_\beta \langle x\rangle^{-1-\varepsilon}.
\end{align}
Hence (\ref{s_a}) is proved by (\ref{s decomposition}), (\ref{s2 estimate}) and (\ref{s1 estimate}).
\end{proof}

\section{Proofs of Lemmas \ref{lem} and \ref{lem2}}
\label{proof of lem}

\subsection{Proof of Lemma \ref{lem}}
First we remark that $J_a, P_\pm, \tilde P_\pm$ and their formal adjoint operators
\begin{align*}
J_a^*u[x]&=(2\pi)^{-d}\int_{\mathbb{T}^d}\sum_{y\in\mathbb{Z}^d} e^{i(x\cdot\xi-\varphi_a(y,\xi))}u[y]d\xi, \\
P_\pm^*u[x]&=(2\pi)^{-d}\int_{\mathbb{T}^d}\sum_{y\in\mathbb{Z}^d} e^{i(x-y)\cdot\xi} p_\pm(x,\xi)u\left[y\right]d\xi, \\
\tilde P_\pm^*u[x]&=(2\pi)^{-d}\int_{\mathbb{T}^d}\sum_{y\in\mathbb{Z}^d} e^{i(\varphi_a(x,\xi)-y\cdot\xi)}p_\pm(x,\xi)u\left[y\right]d\xi
\end{align*} 
map from $\mathscr{S}(\mathbb{Z}^d)$ to itself.

Letting $L:=\langle x-y \rangle^{-2} (1+(x-y)\cdot D_\xi)$, $D_\xi:=\frac{1}{i}\nabla_\xi$, we easily see $L\left(e^{i(x-y)\cdot\xi}\right)=e^{i(x-y)\cdot\xi}$. Thus we have
\begin{align*}
P_\pm u\left[x\right]=&(2\pi)^{-d}\int_{\mathbb{T}^d}\sum_{y\in\mathbb{Z}^d} L^k\left(e^{i(x-y)\cdot\xi}\right) p_\pm(y,\xi)u\left[y\right]d\xi \\
=&(2\pi)^{-d}\int_{\mathbb{T}^d}\sum_{y\in\mathbb{Z}^d} e^{i(x-y)\cdot\xi} (L^*)^k \left(p_\pm(y,\xi)\right)u\left[y\right]d\xi
\end{align*}
for any $k \in \mathbb{N}_{\geq 0}$. We define $|p_\pm|:= \sup_{|\beta|\leq d+1}\sup_{(x,\xi)\in\mathbb{Z}^d\times\mathbb{T}^d} |\partial_\xi^\beta p_\pm(x,\xi)|$. Then we learn that, setting $k=d+1$,
\begin{align*}
|P_\pm u\left[x\right]|\leq C |p_\pm|\sum_{y\in\mathbb{Z}^d} \langle x-y \rangle^{-d-1}|u[x]| .
\end{align*}
This and Young's inequality imply $\|P_\pm u\|\leq C|p_\pm|\| u \|$, where $\| u \| := \left(\sum_{x\in\mathbb{Z}^d} |u[x]|^2 \right)^{\frac{1}{2}}$. Hence $P_{\pm}$ are bounded.

Next we prove $\tilde P_\pm$ are bounded. A direct calculus implies
\begin{align*}
\tilde P_\pm^* \tilde P_\pm u[x]&=(2\pi)^{-d}\int_{\mathbb{T}^d}\sum_{y\in\mathbb{Z}^d} e^{i(\varphi_a(x,\xi)-\varphi_a(y,\xi))}p_\pm(x,\xi)p_\pm(y,\xi)u\left[y\right]d\xi \\
&=(2\pi)^{-d}\int_{\mathbb{T}^d}\sum_{y\in\mathbb{Z}^d} e^{i(x-y)\cdot\eta(\xi;x,y)}p_\pm(x,\xi)p_\pm(y,\xi)u\left[y\right]d\xi,
\end{align*}
where $\eta$ in the last equality is defined by
\begin{align}\label{xi to eta}
\eta(\xi;x,y):=\int_0^1 \nabla_x\varphi_a(y+\theta(x-y),\xi)d\theta .
\end{align}
Then (\ref{phase estimate2}) implies $\eta(\cdot;x,y):\mathbb{T}^d\to\mathbb{T}^d$ has its inverse map $\xi(\cdot;x,y)$. Changing the variable $\xi$ to $\eta$, we have
\begin{align*}
\tilde P_\pm^* \tilde P_\pm u[x]=(2\pi)^{-d}\int_{\mathbb{T}^d}\sum_{y\in\mathbb{Z}^d} e^{i(x-y)\cdot\eta} r(x,y,\eta) u\left[y\right] d\eta,
\end{align*}
where
\begin{align*}
r(x,y,\eta)=p_\pm(x,\xi(\eta;x,y))p_\pm(y,\xi(\eta;x,y))\left|\det\left(\frac{d\xi}{d\eta}\right)\right|.
\end{align*}
Since (\ref{phase estimate}) implies
\begin{align} \label{hensuuhenkan}
\left|\partial_{\eta}^\beta\left[\det\left(\frac{d\xi}{d\eta}\right)-1\right]\right|\leq C_\beta \langle x\rangle^{-\varepsilon} ,
\end{align}
the similar argument for $P_\pm$ proves the boundedness of $\tilde P_\pm^* \tilde P_\pm$. Thus, for $u\in \mathscr{S}
(\mathbb{Z}^d)$, we obtain
\begin{align*}
\|\tilde P_\pm u\|^2=|(\tilde P_\pm^* \tilde P_\pm u,u)|\leq \|\tilde P_\pm^* \tilde P_\pm\|\|u\|^2,
\end{align*}
which implies $\tilde P_\pm$ are bounded. The boundedness of $J_a$ is proved similarly.
\qed

\subsection{Proof of Lemma \ref{lem2}}
Since
\begin{align*}
\gamma(H_0) - P_+ - P_-=\gamma(H_0)(1-\chi),
\end{align*}
the compactness of the support of $1-\chi$ implies $P_+ +P_- -\gamma(H_0)$ is a finite rank operator, in particular, a compact operator.

We show $P_\pm^*-P_\pm$ are compact. We observe
\begin{align*}
&(P_\pm^*-P_\pm)u[x] \\
&=(2\pi)^{-d}\int_{\mathbb{T}^d}\sum_{y\in\mathbb{Z}^d} e^{i(x-y)\cdot\xi} (p_\pm(x,\xi)-p_\pm(y,\xi))u\left[y\right]d\xi \\
&=(2\pi)^{-d}\int_{\mathbb{T}^d}\sum_{y\in\mathbb{Z}^d} e^{i(x-y)\cdot\xi} (x-y)\cdot \int_0^1\nabla_x p_\pm(y+\theta(x-y),\xi)d\theta\: u[y]d\xi \\
&=(2\pi)^{-d}i\int_{\mathbb{T}^d}\sum_{y\in\mathbb{Z}^d} e^{i(x-y)\cdot\xi} \int_0^1\nabla_\xi\cdot\nabla_x p_\pm(y+\theta(x-y),\xi)d\theta\: u[y]d\xi,
\end{align*}
where the last equality follows from integral by parts in $\xi$. Since
\begin{align*}
\left|\int_0^1\partial_\xi^\beta[\nabla_\xi\cdot\nabla_x p_\pm(y+\theta(x-y),\xi)]d\theta\right|&\leq C_\beta\int_0^1 \langle y+\theta(x-y)\rangle^{-1} d\theta \\
&\leq C_\beta^\prime \langle x\rangle^{-1},
\end{align*}
similar argument in Lemma \ref{lem} proves $\langle x\rangle (P_\pm^*-P_\pm)$ are bounded. By the compactness of $\langle x \rangle^{-1}$ as an operator on $\mathcal{H}$, $P_\pm^*-P_\pm=\langle x\rangle^{-1} \cdot \langle x\rangle (P_\pm^*-P_\pm)$ are compact.

We next prove the compactness of $E_\pm(0)-P_\pm$. Using (\ref{xi to eta}), we have
\begin{align*}
E_\pm(0)u[x]&=J_a \tilde P_\pm u[x] \\
&=(2\pi)^{-d}\int_{\mathbb{T}^d}\sum_{y\in\mathbb{Z}^d} e^{i(\varphi_a(x,\xi)-\varphi_a(y,\xi))}p_\pm(y,\xi)u\left[y\right]d\xi \\
&=(2\pi)^{-d}\int_{\mathbb{T}^d}\sum_{y\in\mathbb{Z}^d} e^{i(x-y)\cdot\eta}p_\pm(y,\xi(\eta))\left|\det\left(\frac{d\xi}{d\eta}\right)\right|u\left[y\right]d\eta.
\end{align*}
Thus
\begin{align*}
(E_\pm(0)-P_\pm)u[x]=(2\pi)^{-d}\int_{\mathbb{T}^d}\sum_{y\in\mathbb{Z}^d} e^{i(x-y)\cdot\eta}r(x,y,\eta) u\left[y\right]d\eta,
\end{align*}
where
\begin{align*}
r(x,y,\eta)=p_\pm(y,\xi(\eta))\left|\det\left(\frac{d\xi}{d\eta}\right)\right|-p_\pm(y,\eta).
\end{align*}
By (\ref{hensuuhenkan}), we have $|\partial_\eta^\beta[r(x,y,\eta)]|\leq C_\beta \langle x\rangle^{-\varepsilon}$,
and hence $\langle x\rangle^{\varepsilon}(E_\pm(0)-P_\pm)$ are bounded.
This proves $E_\pm(0)-P_\pm$ are compact.

The compactness of $J_a J_a^*-I$ is proved similarly to that of $E_\pm(0)-P_\pm$, since
\begin{align*}
(J_a J_a^*-I)u[x]&=(2\pi)^{-d}\int_{\mathbb{T}^d}\sum_{y\in\mathbb{Z}^d} e^{i(\varphi_a(x,\xi)-\varphi_a(y,\xi))}u\left[y\right]d\xi-u[x] \\
&=(2\pi)^{-d}\int_{\mathbb{T}^d}\sum_{y\in\mathbb{Z}^d} e^{i(x-y)\cdot\eta}\left(\left|\det\left(\frac{d\xi}{d\eta}\right)\right|-1\right) u\left[y\right] d\eta .
\end{align*}

Finally, we prove $J_a^* J_a-I$ is compact. Now we mimic the proof of Lemma 7.1 in \cite{N2}. For $f\in L^2(\mathbb{T}^d)$, we denote
\begin{align*}
L_a f(\xi)&=F J_a^* J_a F^* f(\xi) \\
&=(2\pi)^{-d}\sum_{x\in\mathbb{Z}^d} \int_{\mathbb{T}^d} e^{i(\varphi_a(x,\xi)-\varphi_a(x,\eta))} f(\eta) d\eta, \\
\tilde L_a f(\xi)&= (2\pi)^{-d}\int_{\mathbb{R}^d} \int_{\mathbb{T}^d} e^{i(\varphi_a(x,\xi)-\varphi_a(x,\eta))} f(\eta) d\eta dx .
\end{align*}

First we show that, for any $\psi\in C^\infty(\mathbb{T}^d)$ with sufficiently small support,
\begin{align*}
K_{a,\psi}:=\psi \circ(L_a-\tilde L_a)
\end{align*}
is a compact operator on $L^2(\mathbb{T}^d)$. We define $\Pi :L^1(\mathbb{R}^d)\to L^1(\mathbb{T}^d)$ by
\begin{align*}
\Pi f(\xi):=\sum_{m\in\mathbb{Z}^d} f(\xi+2\pi m).
\end{align*}
Then (\ref{shuuki}) implies
\begin{align*}
\Pi \tilde L_a f(\xi)=& (2\pi)^{-d}\sum_{m\in\mathbb{Z}^d} \int_{\mathbb{R}^d} \int_{\mathbb{T}^d} e^{i(\varphi_a(x,\xi+2\pi m)-\varphi_a(x,\eta))} f(\eta) d\eta dx \\
=& (2\pi)^{-d}\sum_{m\in\mathbb{Z}^d} \int_{\mathbb{R}^d} \int_{\mathbb{T}^d} e^{i(\varphi_a(x,\xi)+2\pi x\cdot m-\varphi_a(x,\eta))} f(\eta) d\eta dx .
\end{align*}
Using Poisson's summation formula
\begin{align} \label{Poisson}
\sum_{m\in\mathbb{Z}^d} e^{2\pi i x \cdot m}=\sum_{m\in\mathbb{Z}^d}\delta_{x-m}
\end{align}
in the sense of distribution, we have
\begin{align*}
\Pi \tilde L_a f(\xi)=& (2\pi)^{-d}\sum_{x\in\mathbb{Z}^d} \int_{\mathbb{T}^d} e^{i(\varphi_a(x,\xi)-\varphi_a(x,\eta))} f(\eta) d\eta = L_a f(\xi).
\end{align*}
Thus we learn
\begin{align*}
K_{a,\psi}f(\xi)=&\psi \circ (\Pi \tilde L_a-\tilde L_a) f(\xi) \\
=&\sum_{m\in\mathbb{Z}^d \backslash \{0\}} \psi(\xi) \int_{\mathbb{R}^d}\int_{\mathbb{T}^d} e^{i(\varphi_a(x,\xi+2\pi m)-\varphi_a(x,\eta))}f(\eta) d\eta dx \\
=& \int_{\mathbb{T}^d} k_{a,\psi}(\xi,\eta) f(\eta) d\eta ,
\end{align*}
where the integral kernel
\begin{align*}
k_{a,\psi}(\xi,\eta)=\sum_{m\in\mathbb{Z}^d \backslash \{0\}} \psi(\xi) \int_{\mathbb{R}^d} e^{i(\varphi_a(x,\xi+2\pi m)-\varphi_a(x,\eta))} dx
\end{align*}
is smooth. This implies the compactness of $K_{a,\psi}$.

In order to show the compactness of $\psi \circ (\tilde L_a-I)$, we note
\begin{align*}
\tilde L_a f(\xi)= (2\pi)^{-d}\int_{\mathbb{R}^d} \int_{\mathbb{T}^d} e^{i\int_0^1 \nabla_\xi \varphi_a(x,\eta+\theta(\xi-\eta))d\theta \cdot (\xi-\eta)} f(\eta) d\eta dx .
\end{align*}
Letting
\begin{align*}
y(x;\xi,\eta):=\int_0^1 \nabla_\xi \varphi_a(x,\eta+\theta(\xi-\eta))d\theta,
\end{align*}
we observe $y(\cdot ;\xi,\eta)$ has its inverse map by (\ref{phase estimate2}). Thus we have
\begin{align*}
\tilde L_a f(\xi)= (2\pi)^{-d}\int_{\mathbb{R}^d} \int_{\mathbb{T}^d} e^{iy \cdot (\xi-\eta)} \left|\det\left(\frac{dx}{dy}\right)\right| f(\eta) d\eta dy .
\end{align*}
This equality and
\begin{align*}
\left|\partial_y^\alpha \partial_\xi^\beta \partial_\eta^\gamma\left[\det\left(\frac{dx}{dy}\right)-1\right]\right| \leq C_{\alpha\beta\gamma} \langle y \rangle^{-|\alpha|-\varepsilon}
\end{align*}
imply the compactness of $\psi \circ (\tilde L_a-I)$.

Hence, with the help of a partition of unity $\{\psi_j\}_{j=1}^J$ on $\mathbb{T}^d$, we observe
\begin{align*}
J_a^* J_a-I=F^*(L_a-I)F=F^* \sum_{j=1}^J \left( K_{a,\psi_j}+\psi_j \circ (\tilde L_a-I) \right)F
\end{align*}
is compact.
\qed

\end{document}